\documentclass[sigconf,nonacm]{acmart}

\renewcommand\footnotetextcopyrightpermission[1]{}
\usepackage{amsmath}
\usepackage{amsthm}
\usepackage{mathtools}
\usepackage{booktabs}
\usepackage{tabularx}
\usepackage{array}
\usepackage{graphicx}
\usepackage{microtype}
\usepackage{xspace}
\usepackage{cleveref}

\ifPDFTeX
\fi
\ifXeTeX
  
\fi

\providecommand{\tightlist}{%
  \setlength{\itemsep}{0pt}%
  \setlength{\parskip}{0pt}%
}

\newtheorem{proposition}{Proposition}

\title{From Intent to Execution Grant:\\
An Execution-Boundary Conformance Profile for High-Risk AI Actions}

\author{Mengting Wu}
\authornote{Corresponding author.}
\email{chloe@havenlon.com}
\affiliation{%
  \institution{Chengdu Havenlon Security Technology Co., Ltd.}
  \city{Chengdu}
  \country{China}
}

\author{Lin Wang}
\affiliation{%
  \institution{Chengdu Havenlon Security Technology Co., Ltd.}
  \city{Chengdu}
  \country{China}
}

\author{Yong Zhang}
\affiliation{%
  \institution{Chengdu Havenlon Security Technology Co., Ltd.}
  \city{Chengdu}
  \country{China}
}

\author{Jiang Deng}
\affiliation{%
  \institution{Chengdu Havenlon Security Technology Co., Ltd.}
  \city{Chengdu}
  \country{China}
}

\hypersetup{
  pdftitle={From Intent to Execution Grant: An Execution-Boundary Conformance Profile for High-Risk AI Actions},
  pdfauthor={Mengting Wu, Lin Wang, Yong Zhang, Jiang Deng},
  pdfsubject={EBL-Core execution-boundary conformance profile},
  pdfkeywords={AI agents, execution boundaries, authorization, runtime enforcement, formal semantics, conformance profiles}
}

\begin{document}

\begin{abstract}
AI agents increasingly propose externally consequential actions, including financial transfers, infrastructure changes, software deployments, information disclosures, and physical actuation. Authorization engines, policy languages, runtime monitors, provenance mechanisms, and agent guardrails provide important control foundations, but their interfaces do not necessarily define a common semantic contract for the final transition from a particular candidate action to execution authority.

We specify EBL-Core, an execution-boundary conformance profile for determining whether one canonical, fully materialized AI-generated candidate may receive action-scoped execution authority under explicit conditions. EBL-Core relates a structured intent object, a Root Policy, an Operational Policy, root and operational Evidence Obligations, typed evidence, explicit context and time, and a verifiable Decision Derivation. These elements are bound through an Execution Release Contract (ERC), defined as a decision-binding release-condition object. An ERC is not itself an authority-bearing token; a verified \texttt{ALLOW} ERC may support issuance of a separate Execution Grant, whose exercise is governed by Redemption-time validation.

The contribution is the semantic release-and-redemption contract joining these existing mechanism classes, together with conformance requirements for action binding, policy non-weakening, evidence-obligation handling, deterministic adjudication, derivation verification, and grant lifecycle behavior. EBL-Core does not establish human-intent correctness, evidence truth, complete mediation, global non-bypassability, faithful execution, or correct external outcomes; those properties remain conditional on explicit deployment and trust assumptions.

The accompanying minimal reference artifact instantiates one financial-transfer profile with machine-readable schemas, a reference adjudicator, a separately implemented verifier and Semantic Replay path, and a linearizable in-memory grant store. In the retained run, 34 static vectors and 15 lifecycle and mutation checks matched their specified outcomes; 100 trials of 32 concurrent Redemption attempts produced exactly one successful Redemption and one protected test effect per trial, and 100 Revoke--Redeem races ended in a valid terminal outcome. These results demonstrate executability of the specified subset, not production readiness, mechanized correctness, or deployment-level security.
\end{abstract}


\maketitle

\section{Introduction}\label{introduction}

\subsection{Motivation}\label{motivation}

AI systems are increasingly used not only to generate information, recommendations, or plans, but also to propose actions that can change external state. Such actions include initiating financial operations, modifying production infrastructure, deploying software, disclosing protected information, and actuating physical devices. In these settings, an agent's output is not necessarily the terminal result of computation. It may instead become an input to an execution mechanism that possesses credentials, invokes tools, commits transactions, or controls devices.

This transition changes the relevant safety question. For an information-producing system, evaluation can often focus on the quality, correctness, or acceptability of generated content. For an action-producing system, an additional question arises: under what conditions may a particular generated action acquire the authority required to affect an external system?

The distinction is important because an AI-generated action is typically constructed through several stages. A user request may be interpreted into an internal goal, refined into a plan, translated into one or more tool invocations, and finally materialized as a concrete operation containing execution-relevant parameters. These parameters may include a financial amount and recipient, an infrastructure resource and configuration change, a software artifact and deployment environment, a data object and disclosure destination, or a device command and target state. The final action may therefore differ materially from the natural-language request or intermediate plan from which it was derived.

At the same time, the conditions governing an action may change during this process. Evidence may expire, operational state may change, a policy version may be replaced, an approval may cover an earlier object rather than the final one, or the candidate action may be modified after adjudication. A decision that was justified at one stage is not necessarily valid at the moment execution authority is exercised.

These observations motivate an explicit semantic interface between action proposal and externally consequential execution.

\subsection{The Execution-Release Problem}\label{the-execution-release-problem}

Many authorization systems determine whether a principal or structured request is permitted to perform an operation on a resource under applicable policies and contextual attributes. Some can evaluate highly specific requests, including exact operation parameters. These capabilities remain necessary for AI-agent systems: neither an agent nor a supporting service should obtain authority beyond the permissions assigned to the relevant identities, roles, credentials, and policies.

The execution-boundary question specializes this decision rather than replacing it:

\begin{quote}
What minimum semantic release-and-redemption contract must hold before one canonical, fully materialized AI-generated candidate may receive action-scoped execution authority?
\end{quote}

The distinguishing issue is not merely whether an authorization engine can inspect detailed inputs. It is whether the system defines a common contract that binds the final candidate to all decision-relevant conditions and preserves those bindings through the release and Redemption of authority. Depending on the action, this contract may require that:

\begin{itemize}
\tightlist
\item
  the candidate remains within the immutable constraints and permitted refinement dimensions of a trusted intent object;
\item
  the candidate is canonical and fully materialized before adjudication;
\item
  the Root Policy permits the candidate;
\item
  the Operational Policy independently permits the candidate without weakening the Root Policy;
\item
  every obligation in \(Q_K \cup Q_P\) is discharged by evidence classified \texttt{VALID};
\item
  the governing policy, evidence, context, schema, and profile versions are explicit;
\item
  a supplied Decision Derivation verifies under the committed inputs;
\item
  decision-relevant changes trigger denial or re-adjudication; and
\item
  any resulting authority is scoped to the committed candidate and governed by the specified grant lifecycle.
\end{itemize}

We refer to this as the \textbf{execution-release problem}. It concerns the semantic transition from a concrete action proposal to conditions under which action-scoped execution authority may be issued and redeemed.

These requirements can be implemented using existing authorization languages, proof systems, capability mechanisms, and runtime monitors. EBL-Core does not claim that such mechanisms are unable to represent them. Its purpose is to define which bindings must be jointly present for an implementation to conform to a shared execution-release profile.

The distinction also separates semantic adjudication from deployment enforcement. An \texttt{ALLOW} decision and a valid ERC do not by themselves establish that execution is causally dependent on the decision. Complete mediation, correct Grant Issuance, linearized Redemption, faithful realization of the candidate by the Effector, protection of trusted components, and exclusion of alternative authority paths remain deployment properties.

\subsection{Relationship to Existing Abstractions}\label{relationship-to-existing-abstractions}

The execution-release problem builds on several established areas of research and engineering.

Authorization languages define policies over principals, actions, resources, attributes, and contextual inputs. Policy decision points and policy enforcement points separate the evaluation of authorization rules from the application mechanisms that enforce the resulting decisions. These abstractions provide the foundation for representing and evaluating many of the constraints considered in this paper.

Reference monitors characterize the architectural requirements for enforcing a security policy, including complete mediation, tamper resistance, and verifiability. Runtime-assurance architectures similarly place an assured decision component between an untrusted or insufficiently assured controller and a protected system. These approaches establish why a decision mechanism must be positioned outside the authority of the component whose actions it governs.

Agent guardrails and runtime policy systems apply related principles to model-generated actions. They may constrain tool selection, validate arguments, enforce preconditions, restrict privileges, monitor execution histories, or route selected operations through human review. Provenance and audit systems record the origin and evolution of requests, decisions, and execution events. Proof-carrying authorization systems further demonstrate how a decision can be accompanied by evidence that it follows from specified policies and credentials.

These mechanisms address substantial portions of the execution-control problem. The objective of this paper is not to characterize them as generally insufficient, but to identify the integration contract required at a particular boundary. A general authorization engine need not prescribe how natural-language input becomes a trusted intent object. A runtime monitor need not define a portable evidence-obligation representation. A provenance system may describe recorded events without deciding whether a candidate should receive execution authority. Depending on its interface, a guardrail may block a tool invocation without exposing a Decision Derivation that an independent verifier can check. Similarly, an authorization derivation may establish permission from declared premises without specifying how candidate mutation, state change, ERC Generation, Grant Issuance, and Redemption affect the continued usability of that result.

Consequently, two conforming and individually well-designed systems may still disagree about:

\begin{itemize}
\tightlist
\item
  which object was authorized;
\item
  which evidence was required;
\item
  which policy versions governed the decision;
\item
  whether an earlier approval covers a refined candidate;
\item
  when a decision becomes stale;
\item
  what authority an \texttt{ALLOW} result releases; and
\item
  what an independent verifier must reconstruct.
\end{itemize}

The residual abstraction considered here is a shared contract for these execution-release semantics. Such a contract can be implemented using existing policy languages, proof systems, reference monitors, and agent runtimes. It does not require replacing them.

\subsection{Execution-Boundary Semantics}\label{execution-boundary-semantics}

We define an \textbf{execution boundary} as the last scoped mediation point at which a proposed action can still be denied before an externally relevant effect occurs. The definition is scoped to declared action classes and execution paths. Whether every path capable of producing a protected effect actually traverses the boundary is a deployment property.

EBL-Core specifies the semantic behavior required at this boundary. Its adjudication function has the form

\[
\Gamma_K(P,I,x,Ev,Ctx,t)\rightarrow(d,r,\pi),
\]

where \(I\) is a trusted intent object; \(x\) is one canonical, fully materialized candidate action; \(K\) is the versioned Root Policy; \(P\) is the versioned Operational Policy; \(Ev\) is the complete materialized evidence input; \(Ctx\) is explicit decision-relevant context; \(t\) is explicit adjudication time; \(d\) is the decision; \(r\) is a stable reason code; and \(\pi\) is a Decision Derivation.

The applicable Evidence Obligations are derived separately:

\[
Q_K=Obligations_K(K,I,x,Ctx)
\]

and

\[
Q_P=Obligations_P(P,I,x,Ctx).
\]

The Operational Policy may add obligations through \(Q_P\), but it cannot remove, weaken, rename, or reinterpret obligations in \(Q_K\). This requirement is part of \textbf{root-policy dominance}. It is distinct from conjunctive rule addition and from the compatibility relation used to classify Operational Policy evolution.

The candidate, rather than an earlier plan or textual request, is the immediate object of adjudication. \texttt{Bound(I,x,Ctx,t)} permits explicitly authorized refinement but requires the final candidate to preserve immutable intent constraints. This relation establishes consistency with the trusted intent object; it does not establish that the object correctly represents a human's latent intent.

For a positive decision, every applicable obligation in \(Q_K\cup Q_P\) must resolve to \texttt{VALID}. \texttt{UNKNOWN}, \texttt{MISSING}, \texttt{EXPIRED}, and \texttt{CONFLICT} remain distinct diagnostic states, but none discharges a positive obligation. Evidence classification concerns the declared evidence interface and does not establish external truth.

Adjudication is connected to authority through four distinct operations:

\begin{verbatim}
Adjudicate(B) -> (d, r, pi)
GenerateERC(B, d, r, pi, ...) -> erc | bottom
IssueGrant(erc, B, pi) -> g | bottom
Redeem(g, erc, B, pi, R) -> EXECUTE | DENY
\end{verbatim}

ERC Generation creates a decision-binding release-condition object. The ERC commits to the adjudicated inputs, obligations, decision, reason code, validity conditions, and supplied Decision Derivation, but is not inherently authority-bearing. Grant Issuance may release action-scoped execution authority only from a verified \texttt{ALLOW} ERC. Under the baseline profile, the grant begins in \texttt{ISSUED} and a successful, linearized Redemption changes it to \texttt{CONSUMED}; \texttt{CONSUMED}, \texttt{EXPIRED}, and \texttt{REVOKED} are terminal for that grant instance.

EBL-Core defines semantic properties required of conforming adjudicators, ERC generators, Grant Issuers, verifiers, and Redemption Interfaces. Stronger claims about external effects require the designated grant to be necessary for the scoped action, the protected transition and grant-state update to be linearized, the Effector to realize the committed candidate faithfully, and alternative effect-producing paths to be excluded or separately governed.

\subsection{Contributions}\label{contributions}

This paper makes three contributions.

\begin{enumerate}
\def\labelenumi{\arabic{enumi}.}
\item
  \textbf{An execution-boundary conformance model.}\\
  We characterize the execution-release problem and define the typed actors, objects, lifecycle stages, trust assumptions, and deployment boundary required to connect one canonical AI-generated candidate to action-scoped execution authority.
\item
  \textbf{Intent-bound, evidence-aware, non-weakening adjudication semantics.}\\
  We specify adjudication over trusted intent, one canonical candidate, Root Policy, Operational Policy, separately derived obligations \(Q_K\) and \(Q_P\), typed evidence, explicit context, and explicit time. The profile separates root-policy dominance, conjunctive rule addition, and policy-version evolution, and requires deterministic, side-effect-free, input-closed, and bounded evaluation.
\item
  \textbf{The Execution Release Contract.}\\
  We define the ERC as a decision-binding release-condition object connecting adjudication to separate Grant Issuance and Redemption operations. It commits to the exact candidate and decision state, including a supplied Decision Derivation, and provides a common verification and conformance target without defining a new general-purpose capability primitive.
\end{enumerate}

Together, these contributions define an execution-release profile that can be implemented over heterogeneous authorization, evidence, capability, and runtime mechanisms. The included artifact validates a bounded transfer profile; interoperability among independently developed implementations remains an empirical question for future conformance studies.

\section{Related Work}\label{related-work}

This section organizes related work by the abstraction each research family treats as primary. The purpose is not to argue that existing mechanisms cannot implement EBL requirements. In many cases, they can and should serve as EBL backends. The relevant distinction is between capabilities that a system may provide and obligations that its primary abstraction specifies by default.

EBL does not claim that policy evaluation, capability security, proof-carrying authorization, runtime mediation, or agent action control is new. Its proposed residual is a common execution-release profile that specifies which versioned facts must jointly justify an action-scoped execution grant, how that grant is bound to the final candidate action, and when the grant must be rejected or re-adjudicated.

\subsection{Authorization Languages}\label{authorization-languages}

Authorization languages define how policies determine whether an identified principal may perform an action on a resource under a given set of attributes and environmental conditions. Attribute-based access control provides the general principal--operation--object--environment foundation on which these policy engines build~\cite{nist-sp800-162}. They provide the most direct foundation for EBL's policy-evaluation semantics.

Cedar~\cite{policy-cedar-2024} defines an expressive and analyzable authorization language organized around principal, action, resource, and context. Its design includes schema-based validation, explicit permit and forbid policies, default denial, and formal modeling suitable for policy analysis. Cedar demonstrates that a policy language can combine practical evaluation performance with a semantics precise enough to support equivalence and authorization analysis.

Open Policy Agent and Rego~\cite{policy-opa-rego} provide a domain-independent policy-decision mechanism over structured input. Rego separates policy evaluation from application logic and permits applications to request structured decisions from an external or embedded policy engine. This separation allows the same policy framework to govern API access, infrastructure configuration, deployment controls, and other application-specific decisions.

XACML~\cite{policy-xacml-3} defines a standardized attribute-based access-control architecture involving policy administration, policy decision, policy information, and policy enforcement points. It also defines \texttt{Permit}, \texttt{Deny}, \texttt{NotApplicable}, and several forms of \texttt{Indeterminate}, together with policy-combining algorithms, obligations, and advice.

\emph{Owned abstraction.}\\
These systems own the abstraction of a policy decision over structured authorization inputs. Their central question is whether a request, usually represented through subject, action, resource, and environment attributes, is permitted under a policy set.

\emph{What EBL inherits.}\\
EBL inherits declarative policy evaluation, explicit decision inputs, schema or type validation, default-denial behavior, policy-combining semantics, separation of decision logic from application code, and the treatment of unavailable information as distinct from positive authorization. EBL's root and operational policies need not be evaluated by a new engine; a conforming implementation could compile them to Cedar or Rego or express them through an XACML profile.

\emph{What remains unspecified.}\\
Within the base abstractions and documented integrations considered here, authorization languages do not uniformly require one execution-release contract relating trusted intent, one canonical candidate, Evidence Obligations, policy versions, Decision-Derivation Verification, Grant Issuance, and Redemption. Individual requirements may be expressed using attributes, policies, obligations, or application logic, but their cross-component semantics depend on the integration.

In particular, a policy decision does not necessarily specify:

\begin{itemize}
\tightlist
\item
  how a trusted intent constrains the candidate's permitted refinements;
\item
  whether the decision covers the exact canonical candidate submitted for Redemption;
\item
  how \(Q_K\) and \(Q_P\) are derived and protected from operational weakening;
\item
  how evidence freshness, absence, or conflict affects those obligations;
\item
  which candidate, policy, evidence, context, and time changes invalidate the result;
\item
  which Decision Derivation an independent verifier must accept;
\item
  whether a decision may support Grant Issuance; or
\item
  how grant state and current conditions are checked during Redemption.
\end{itemize}

EBL-Core does not claim greater general policy expressiveness. It makes these conditions mandatory within a restricted execution-release conformance profile.

\subsection{Capability and Proof-Carrying Authorization}\label{capability-and-proof-carrying-authorization}

Capability-based systems~\cite{cap-dennis-vanhorn-1966} represent authority through protected references or tokens that designate both an object and the operations permitted on it. Possession of an appropriate capability is a prerequisite for exercising the associated authority. Capability systems support least authority by allowing rights to be scoped, delegated, attenuated, and, depending on the design, revoked or time-bounded. Macaroons illustrate how contextual caveats can be bound to decentralized authorization credentials without prescribing EBL's candidate and Redemption lifecycle~\cite{cap-macaroons-2014}.

Proof-Carrying Authorization extends authorization with machine-checkable evidence that a request follows from a set of policies and credentials. In a typical design, an untrusted requester supplies credentials and a logical derivation, while a trusted verifier checks whether the derivation establishes the requested authorization. The Proof-Carrying Authorization System~\cite{proof-pca-2001} establishes this general architecture, while systems such as the Proof-Carrying File System~\cite{proof-pcfs-2010} connect proof verification to dynamic policies and conditional capabilities.

\emph{Owned abstraction.}\\
Capability systems own the representation and controlled exercise of authority. Proof-carrying authorization owns the derivation and verification of an authorization conclusion from declared policies, credentials, and trusted roots.

\emph{What EBL inherits.}\\
EBL inherits the principle that authorization should be represented by an object whose scope can be independently checked, rather than by an informal statement that an action was approved. It also inherits the distinction between proof construction and proof verification, the use of explicit trusted roots, and the requirement that a protected operation depend on successful verification.

The EBL decision derivation is consequently not proposed as a new proof-carrying primitive. It is an execution-specific derivation whose obligations may be implemented using established authorization logics or proof systems.

\emph{What remains unspecified.}\\
General capability and proof-carrying authorization models do not, by themselves, standardize the lifecycle through which an AI-generated proposal becomes a final candidate action. They do not necessarily prescribe:

\begin{itemize}
\tightlist
\item
  a trusted representation of the request-level intent;
\item
  a refinement relation between that intent and the materialized action;
\item
  commitments to the final payload and its execution-relevant parameters;
\item
  typed evidence obligations and their freshness conditions;
\item
  invalidation after policy, evidence, context, or candidate mutation;
\item
  a common execution-release certificate format; or
\item
  a redemption-time check against the current decision-relevant state.
\end{itemize}

These properties can be represented in sufficiently expressive authorization logics. EBL's proposed contribution is to make them mandatory parts of a specific conformance profile.

The distinction is also epistemic. A proof-carrying authorization mechanism may establish that an authorization conclusion follows from supplied policies, credentials, and premises. Within EBL-Core, successful Decision-Derivation Verification similarly establishes only that the recorded decision follows under the declared profile semantics, committed inputs, and trust assumptions. It does not establish that the evidence premises accurately describe the external world or that the authorized effect subsequently occurred.

\subsection{Reference Monitor and Runtime Assurance}\label{reference-monitor-and-runtime-assurance}

A reference monitor~\cite{runtime-reference-monitor} is an architectural abstraction for policy enforcement. Its reference-validation mechanism is expected to mediate every relevant access, resist tampering, and remain sufficiently small and well specified to support analysis. Complete mediation requires authorization to be reconsidered on every relevant access because decision state may change~\cite{saltzer-schroeder-1975}; execution monitoring further distinguishes policies enforceable from observed event prefixes~\cite{schneider-2000}. The reference-monitor concept therefore addresses where enforcement authority must reside and what structural properties the enforcement mechanism must possess.

Runtime Assurance and Simplex architectures~\cite{runtime-simplex-1996,runtime-assurance-2021} apply a related pattern to safety-critical control. An unverified advanced controller may operate while a trusted monitor determines whether its behavior remains within a safety envelope. When the relevant safety condition can no longer be maintained, control is transferred to an assured or reduced-capability controller. These architectures separate high-performance but insufficiently assured decision-making from the component responsible for preserving system safety.

\emph{Owned abstraction.}\\
Reference monitors own complete mediation and protected policy enforcement. Runtime-assurance systems own the supervisory relationship between an untrusted controller, a safety monitor, and an assured fallback or recovery controller.

\emph{What EBL inherits.}\\
EBL inherits the separation between an untrusted action proposer and a trusted enforcement decision. A model, planner, tool selector, or policy generator may propose an action or a policy update, but it is not the final authority for releasing action-scoped execution authority.

EBL also inherits fail-closed mediation and the principle that a reduced operating mode should be represented as an explicitly restricted capability regime rather than as an ambiguous authorization result. If a deployment supports a safe mode, its available actions and transition rules must be specified independently from the Boolean decision concerning an individual candidate action.

\emph{What remains unspecified.}\\
A reference monitor is parametrized by the policy it enforces. The abstraction does not determine which intent, approval, evidence, policy-version, or action-identity obligations should govern an AI-generated transaction. Runtime-assurance frameworks commonly focus on whether a controller's proposed behavior preserves a state-space safety condition. They do not necessarily define the authorization and evidence semantics required for heterogeneous digital actions involving identities, recipients, assets, infrastructure resources, disclosure destinations, or externally issued approvals.

Conversely, EBL semantics do not establish that an implementation possesses reference-monitor properties. EBL-Core can require an Execution Grant to remain bound to a verified Decision Derivation and committed decision state, but those semantic requirements do not establish that every effect-producing path requires the grant. Complete mediation, protection of the verifier, exclusive control of the necessary capability, and exclusion of alternative execution paths remain deployment assumptions.

Usage-control models provide a further foundation by treating authorizations, obligations, mutable attributes, and continuing decisions during usage as first-class concepts~\cite{park-sandhu-2004}. EBL inherits the need to reconsider current conditions but narrows its scope to a candidate-bound release and single-use Redemption contract. The relationship is therefore complementary: reference-monitor and usage-control work define structural and continuing enforcement concerns, while EBL defines the proposed execution-release decision profile that such mechanisms may enforce.

\subsection{AI Agent Action Control}\label{ai-agent-action-control}

Recent agent-security systems place enforcement closer to model-generated tool calls and action proposals. This work establishes that action-time mediation for agents is already an active research area.

AgentSpec~\cite{agent-agentspec-2025} introduces a domain-specific language for specifying runtime constraints through triggers, predicates, and enforcement mechanisms. It demonstrates that agent behavior can be checked against structured rules across code execution, embodied-agent, and autonomous-driving scenarios.

Progent~\cite{agent-progent-2025} represents agent privilege through symbolic policies over tool names and arguments. Every tool call is checked through a deterministic procedure. Progent also distinguishes narrowing policy updates from privilege expansions, using an SMT solver to permit automatic narrowing while requiring separate approval for expansion. This makes monotonic confinement an explicit part of the agent-policy lifecycle.

ToolGate~\cite{agent-toolgate-2026} models tools using preconditions and postconditions over a typed symbolic state. Preconditions control invocation, while postconditions determine whether tool results may update the trusted symbolic state. ToolGate therefore treats tool execution as a contract-governed state transition rather than as an unconstrained continuation of model reasoning.

Intent-Governed Access Control~\cite{agent-igac-2026} converts a trusted request into a short-lived intent certificate, narrows the statically authorized tool manifest, and checks proposed tool and payload effects before execution. IGAC explicitly distinguishes static-policy non-expansion from request-level confinement and treats the model, intent classifier, and planner as non-authoritative components.

FORGE~\cite{agent-forge-2026} treats policy enforcement as a cross-cutting concern independent of agent reasoning. It uses Datalog policies, a reference monitor, and an observability service governed by an assume--guarantee contract. Its policies may depend on causal execution history and may be enforced across multiple agents at policy-relevant actions.

Atomic Decision Boundaries~\cite{boundary-adb-2026} makes the relationship between decision validity and protected state transition explicit. Its primary abstraction is the atomic or linearized boundary needed to avoid a gap between checking decision-relevant state and applying the corresponding effect. EBL-Core inherits this requirement for Redemption; it does not claim to originate atomic decision-effect coupling.

Proof-Carrying Agent Actions~\cite{boundary-pcaa-2026} introduces portable action certificates and runtime verification for agent actions. Its primary abstraction is a proof-carrying action artifact that can accompany an action across system boundaries. EBL-Core does not claim the first portable action certificate or proof-carrying agent action. Its residual concerns the particular decision inputs, policy and obligation separation, ERC role, grant lifecycle, and Redemption semantics required by its conformance profile.

\emph{Owned abstraction.}\\
AgentSpec owns a runtime constraint language for agent behavior; Progent owns symbolic tool-policy enforcement and monotonic privilege management; ToolGate owns typed precondition and postcondition checking for tool-mediated state transitions; IGAC owns request-derived intent certification and payload-level confinement; FORGE owns formal, history-aware runtime policy enforcement; Atomic Decision Boundaries owns decision-effect coupling and linearization; and Proof-Carrying Agent Actions owns portable action certificates and runtime verification.

\emph{What EBL inherits.}\\
EBL-Core inherits external mediation of model-generated actions, typed tool and candidate representations, deterministic policy evaluation, request-scoped confinement, non-expanding policy changes, history- and context-sensitive predicates, portable verification artifacts, and linearized checking at the protected transition. These systems establish substantial portions of the mechanism space on which EBL-Core builds.

\emph{What remains unspecified.}\\
Across the systems compared here, the properties required by EBL-Core are distributed rather than uniformly specified by one shared profile. The residual addressed by EBL-Core is the joint contract among:

\begin{itemize}
\tightlist
\item
  one canonical, fully materialized candidate;
\item
  a trusted intent object and its permitted refinement relation;
\item
  a Root Policy and Operational Policy;
\item
  separately identified \(Q_K\) and \(Q_P\);
\item
  complete typed evidence and explicit context;
\item
  a verifiable Decision Derivation;
\item
  an ERC that is not inherently authority-bearing;
\item
  separate ERC Generation and Grant Issuance;
\item
  the \texttt{ISSUED}, \texttt{CONSUMED}, \texttt{EXPIRED}, and \texttt{REVOKED} grant lifecycle; and
\item
  single-use, linearized Redemption at the protected transition.
\end{itemize}

This residual is narrower than discovering action-time mediation, intent certificates, action certificates, deterministic tool gates, atomic boundaries, or proof-carrying authorization. EBL-Core instead defines the contents and lifecycle that an implementation must expose to claim conformance with this particular execution-release profile.

Provenance and remote-attestation architectures address another complementary interface. W3C PROV models entities, activities, agents, and derivation relations~\cite{w3c-prov-dm}; RATS separates an Attester, Verifier, and Relying Party when evidence is appraised~\cite{rfc9334}. EBL-Core can consume records produced by such systems, but its \texttt{VALID} status remains an obligation-relative adjudication result rather than a general provenance or attestation truth claim.

\subsection{The Execution-Release Gap}\label{the-execution-release-gap}

The preceding work establishes the component abstractions on which EBL-Core relies. Table~\ref{tab:abstraction-matrix} maps six close baselines to the joint obligations of the profile. The classification concerns what each cited abstraction specifies directly, not what a sufficiently programmable implementation could encode. \textbf{N} denotes native semantics, \textbf{A} an explicit adapter convention, \textbf{E} an external component, and \textbf{U} a property not established by the cited abstraction. An \textbf{A}, \textbf{E}, or \textbf{U} entry is not a quality judgment.

\begin{table*}[t]
\centering
\scriptsize
\setlength{\tabcolsep}{3pt}
\renewcommand{\arraystretch}{1.16}
\begin{tabularx}{\textwidth}{@{}>{\raggedright\arraybackslash}p{0.20\textwidth}*{6}{>{\centering\arraybackslash}X}@{}}
\toprule
\textbf{Primary abstraction} & \textbf{Intent--candidate} & \textbf{Policy--obligation split} & \textbf{Verifiable derivation} & \textbf{Contract--authority split} & \textbf{Grant lifecycle} & \textbf{Linearized Redemption} \\
\midrule
Cedar~\cite{policy-cedar-2024} & A & A & U & E & E & E \\
Proof-Carrying Authorization~\cite{proof-pca-2001,proof-pcfs-2010} & A & A & N & A & E & E \\
Progent~\cite{agent-progent-2025} & E & A & U & U & E & E \\
Intent-Governed Access Control~\cite{agent-igac-2026} & N & A & U & U & E & E \\
Atomic Decision Boundaries~\cite{boundary-adb-2026} & E & E & E & E & E & N \\
Proof-Carrying Agent Actions~\cite{boundary-pcaa-2026} & A & E & N & U & E & E \\
\midrule
\textbf{EBL-Core} & \textbf{N} & \textbf{N} & \textbf{N} & \textbf{N} & \textbf{N} & \textbf{N} \\
\bottomrule
\end{tabularx}
\caption{Conservative mapping of documented primary abstractions to EBL-Core obligations. The table distinguishes semantic ownership from expressibility and does not rank implementation strength or performance.}
\label{tab:abstraction-matrix}
\end{table*}

No comparison row is expected to reproduce EBL-Core because each baseline owns a different abstraction. The residual is the mandatory composition of one canonical candidate, trusted-intent binding, non-weakening policy and obligation separation, typed evidence resolution, a verifiable Decision Derivation, an ERC distinct from released authority, and single-use linearized Redemption. Cedar, Rego, or XACML may evaluate predicates; a proof system may encode a derivation; a capability system may represent the grant; and a reference monitor may govern Redemption. EBL-Core defines the observable joint contract by which that composition can claim conformance.

\section{System Model and Execution Release Contract}\label{system-model-and-execution-release-contract}

\subsection{System Overview}\label{system-overview}

We consider an AI-assisted system in which an agent may propose an action whose execution can produce an externally consequential state transition. The agent does not obtain execution authority merely by generating a syntactically valid action. Instead, the proposed action passes through logically distinct stages:

\begin{verbatim}
Intent Establishment
        |
Candidate Action Materialization
        |
Adjudication
        |
ERC Generation
        |
Grant Issuance
        |
Grant Redemption
        |
External Effect
\end{verbatim}

These stages define semantic responsibilities rather than a required deployment topology. A conforming system may colocate several stages within one service or distribute them across multiple trust domains. Properties that depend on separation remain conditional on the deployment.

\textbf{Intent establishment} converts an authenticated instruction into a structured intent object \(I\). The object identifies the authority under which an action may be considered and the constraints that must survive subsequent refinement. EBL-Core begins after this object has been established; it does not define how natural-language requests are interpreted.

\textbf{Candidate-action materialization} produces one concrete candidate \(x\). Every decision-relevant parameter required to identify the proposed operation must have a committed value. A plan, operation class, partial payload, unresolved parameter, or later-selected target is not a materialized EBL-Core candidate.

\textbf{Adjudication} evaluates the candidate against the intent object, the applicable root-policy version, the operational-policy version, their separately derived evidence obligations, explicit context, and explicit time. It returns a decision, a reason, and a decision derivation.

\textbf{ERC generation} binds the adjudication result to its decision inputs and release conditions. The resulting Execution Release Contract is a decision-binding release-condition object. It is not inherently authority-bearing.

\textbf{Grant issuance} releases action-scoped execution authority from a verified \texttt{ALLOW} ERC. Issuance creates a grant in the \texttt{ISSUED} state. A \texttt{DENY} ERC cannot support grant issuance.

\textbf{Grant redemption} validates whether the issued authority may be exercised for the exact candidate under current decision-relevant conditions. The baseline profile permits one successful redemption. Successful redemption changes the grant state from \texttt{ISSUED} to \texttt{CONSUMED}.

\textbf{External effect} is the protected state transition produced through the governed interface. Let \(S\) denote protected state and let

\[
\Delta:S\times X\rightharpoonup S
\]

be the partial transition function for candidate actions. EBL-Core requires validation, single-use grant consumption, and the protected transition to be linearized with respect to decision-relevant state. It specifies this semantic requirement without prescribing an implementation mechanism.

The resulting path is:

\[
I\longrightarrow x
\longrightarrow\Gamma_K(P,I,x,Ev,Ctx,t)
\longrightarrow erc
\longrightarrow g
\longrightarrow\Delta(s,x).
\]

An adjudication result is not an ERC; an ERC is not an execution grant; successful grant redemption does not establish that the external action completed correctly or produced its intended outcome.

\subsection{Semantic Objects}\label{semantic-objects}

Let the semantic domains be:

\[
\begin{aligned}
I &\in Intent,\\
x &\in Candidate,\\
Ev &\in EvidenceSet,\\
P &\in OperationalPolicy,\\
K &\in RootPolicy,\\
Ctx &\in Context,\\
t &\in Time,\\
erc &\in ExecutionReleaseContract,\\
g &\in Grant.
\end{aligned}
\]

We additionally use:

\[
d\in\{ALLOW,DENY\},\qquad
r\in Reason,\qquad
\pi\in Derivation,
\]

and

\[
state(g)\in
\{ISSUED,CONSUMED,EXPIRED,REVOKED\}.
\]

For a semantic object \(o\), let

\[
id_v(o)=Commit_v(Canon_v(o))
\]

denote its identity commitment under profile version \(v\). We omit the subscript on \(id\) where \(v\) is fixed. \(Canon_v\) is a deterministic canonical representation, and \(Commit_v\) is an abstract commitment mechanism selected by the applicable profile. A concrete profile may instantiate these operations using a specified canonical encoding and collision-resistant digest, such as a JSON canonicalization scheme and a named hash function~\cite{rfc8785}.

\subsubsection{Intent}\label{intent}

An intent object \(I\) is a structured input established by an identified Intent Authority:

\[
\begin{aligned}
I=\langle{}&subject,purpose,scope,constraints,\\
&refinement,authority,validity\rangle.
\end{aligned}
\]

The relation

\[
Bound(I,x,Ctx,t)
\]

holds when the exact candidate \(x\) is an authorized materialization of \(I\) under context \(Ctx\) at time \(t\). This relation concerns conformance to the structured intent object. It does not establish that \(I\) correctly represents latent or natural-language human intent.

\subsubsection{Candidate Action}\label{candidate-action}

A candidate action is one fully materialized proposed operation:

\[
\begin{aligned}
x=\langle{}&actor,operation,target,parameters,\\
&declaredEffects,authorityScope\rangle.
\end{aligned}
\]

\texttt{Complete(x)} holds only when every decision-relevant parameter has one committed value. A parameter domain permitted by \(I\) must therefore be resolved to a concrete value before adjudication. If a target, payload, recipient, artifact, or other effect-relevant value is selected or changed later, the result is a different candidate requiring a new adjudication.

Two candidates are equivalent only when their canonical decision-relevant representations are equivalent:

\[
x_1\equiv x_2
\iff
Canon_v(x_1)=Canon_v(x_2).
\]

The applicable candidate profile must identify every field that may affect the governed operation or its declared effect. The baseline EBL-Core profile does not authorize a set of alternative candidates through one ERC. Set-valued or bounded-action authorization is reserved for a future extension profile.

Candidate completeness and declared-effect binding do not establish that the Effector will implement the declared effect faithfully. Effector correctness remains a deployment assumption.

\subsubsection{Evidence}\label{evidence}

The submitted evidence set is finite:

\[
Ev=\{e_1,\ldots,e_n\}.
\]

Root and operational evidence obligations are generated separately:

\[
Q_K=Obligations_K(K,I,x,Ctx),
\]

\[
Q_P=Obligations_P(P,I,x,Ctx),
\]

and composed as:

\[
Q=Q_K\cup Q_P.
\]

Each obligation identifies its policy origin and the versioned resolution semantics under which it is evaluated. Operational policy \(P\) cannot remove, weaken, rename, or reinterpret an obligation in \(Q_K\). It may introduce additional obligations through \(Q_P\).

For each \(q\in Q\),

\[
\begin{aligned}
Status(q,Ev,Ctx,t)\in\{&VALID,UNKNOWN,MISSING,\\
&EXPIRED,CONFLICT\}.
\end{aligned}
\]

Only \texttt{VALID} discharges a positive obligation:

\[
Discharged(q,Ev,Ctx,t)
\iff
Status(q,Ev,Ctx,t)=VALID.
\]

These values classify obligation resolution under the declared evidence profile. They do not assign context-free truth values to external assertions.

\subsubsection{Root Policy}\label{root-policy}

The root policy \(K\) defines non-overridable constraints for one identified policy regime and version. It is fixed for a particular adjudication and recorded in the ERC. Root-policy dominance does not imply that \(K\) is correct, complete, or immutable across administrative updates.

Operational policy cannot change:

\begin{itemize}
\tightlist
\item
  the semantics of \(Permit_K\);
\item
  the contents or interpretation of \(Q_K\);
\item
  the canonicalization rules applied to root-policy objects; or
\item
  the root-policy version recorded for the decision.
\end{itemize}

A change to \(K\) produces a different adjudication input and invalidates baseline reuse of the earlier ERC.

Within a profile, a policy-version identifier is immutable and content-bound: it denotes one canonical policy object and its declared evaluation semantics. Any semantic change requires a new identifier; reuse after a content change is non-conformant.

\subsubsection{Operational Policy}\label{operational-policy}

The operational policy \(P\) contains mutable authorization, workflow, and deployment restrictions. It may deny additional candidates or generate additional evidence obligations.

Define:

\[
Permit_K(I,x,Ev,Ctx,t)
\]

and

\[
Permit_P(I,x,Ev,Ctx,t).
\]

The combined direct policy condition is conjunctive:

\[
Permit_{K,P}
=
Permit_K\land Permit_P.
\]

This conjunction is evaluated together with the separately composed obligation set \(Q_K\cup Q_P\). Root-policy dominance follows only under the constraint that \(P\) cannot modify the root-policy predicate or root obligations.

\subsubsection{Context and Time}\label{context-and-time}

\(Ctx\) is the finite, explicit decision-relevant context supplied to adjudication. Upstream state may be projected into \(Ctx\), but the resulting projection must contain every contextual field on which the decision depends.

Time \(t\) is a separate explicit input. Evidence freshness, intent validity, policy validity, and other time-dependent predicates are evaluated against \(t\). Trust in the supplied time source is a deployment assumption.

\subsubsection{Adjudication}\label{adjudication}

Define the complete adjudication input bundle:

\[
B=
\langle
profileVersion,K,P,I,x,Ev,Ctx,t
\rangle.
\]

The adjudication function is:

\[
\Gamma_K(P,I,x,Ev,Ctx,t)
\rightarrow
(d,r,\pi).
\]

For well-formed inputs, \texttt{ALLOW} is produced only if:

\[
\begin{aligned}
&WellFormed(B)\\
{}\land{}&Bound(I,x,Ctx,t)\\
{}\land{}&Permit_K(I,x,Ev,Ctx,t)\\
{}\land{}&Permit_P(I,x,Ev,Ctx,t)\\
{}\land{}&\forall q\in Q_K:
Status(q,Ev,Ctx,t)=VALID\\
{}\land{}&\forall q\in Q_P:
Status(q,Ev,Ctx,t)=VALID.
\end{aligned}
\]

The decision derivation \(\pi\) binds the complete input bundle, the separately generated obligation sets, the evidence classifications, the applied inference steps, and the final decision and reason.

\subsubsection{Execution Release Contract}\label{execution-release-contract}

An ERC binds an adjudication result to the conditions under which action-scoped authority may be issued and redeemed. An ERC concerns exactly one canonical candidate action.

\subsection{Execution Release Contract}\label{execution-release-contract-1}

Conceptually, an ERC is:

\[
\begin{aligned}
erc=\langle
&ercVersion,profileVersion,id(B),\\
&id(I),id(x),authorityScope,\\
&version(K),version(P),id(Q_K),id(Q_P),id(Ev),\\
&id(Ctx),t,interval,d,r,id(\pi),issuer,nonce
\rangle.
\end{aligned}
\]

Here \(t\) is the explicit adjudication time already contained in \(B\); redemption time is written separately as \(t_r\). The ERC may contain commitments rather than complete embedded objects, provided that a verifier can obtain the bound objects through a defined verification context.

An ERC is a decision-binding release-condition object. It is not inherently an authority-bearing token; it specifies the conditions under which authority may be released. A deployment may embed an ERC, an ERC commitment, or an ERC reference in a capability or credential format, but the following semantic roles remain distinct:

{\small
\begin{verbatim}
ERC:
    Which decision and release conditions were established?

Execution Grant:
    Which action-scoped authority was actually released?

Redemption:
    May that authority be exercised now?
\end{verbatim}
}

A nonce distinguishes an ERC or issuance instance and supports identity and correlation. It does not, by itself, enforce single use or prevent replay. Replay prevention requires grant-state consumption or an equivalent linearizable mechanism.

\subsubsection{ERC versus an Audit Record}\label{erc-versus-an-audit-record}

An audit record describes an event for later inspection. An ERC is evaluated prospectively as part of grant issuance or redemption. It may subsequently be retained in an audit trail, but its existence does not establish that:

\begin{itemize}
\tightlist
\item
  a grant was issued;
\item
  the grant was redeemed;
\item
  the Effector executed the candidate; or
\item
  the expected external outcome occurred.
\end{itemize}

\subsubsection{ERC versus a Decision-Derivation Record}\label{erc-versus-a-decision-derivation-record}

A decision-derivation record explains how an evaluator derived a decision. The ERC additionally binds that derivation to:

\begin{itemize}
\tightlist
\item
  the exact intent and candidate;
\item
  the applicable authority scope;
\item
  root and operational policy versions;
\item
  root and operational evidence obligations;
\item
  the submitted evidence;
\item
  decision-relevant context and time;
\item
  validity conditions;
\item
  issuer identity; and
\item
  grant-release conditions.
\end{itemize}

The ERC commits to the supplied derivation through \(id(\pi)\). An independently generated derivation need not have the same serialization, provided that it verifies against the same input bundle and yields the same semantic decision and reason.

\subsubsection{Action Binding}\label{action-binding}

An ERC is action-bound when:

\[
erc.candidateId=id(x).
\]

A grant issued from the ERC may be redeemed only for that candidate:

\[
id(x')\neq erc.candidateId
\implies
Redeem(g,erc,x',\ldots)=DENY.
\]

The authority representation may narrow the permissions needed to execute \(x\), but it cannot substitute another candidate. Authorization of a candidate set is outside the EBL-Core baseline.

\subsubsection{State and Policy Binding}\label{state-and-policy-binding}

The adjudication fingerprint is:

\[
\begin{aligned}
F(B)=Commit(&
profileVersion,id(I),id(x),\\
&version(K),version(P),id(Q_K),id(Q_P),\\
&id(Ev),id(Ctx),t).
\end{aligned}
\]

The ERC records \(F(B)\), the component commitments needed to reconstruct it, or both.

The baseline profile requires the root- and operational-policy versions at redemption to match the versions used at adjudication. A policy-version change invalidates reuse of the earlier ERC and requires re-adjudication.

A separately established policy-evolution compatibility relation may classify a new operational policy as no more permissive than an earlier one. That classification does not, by itself, preserve a previously issued \texttt{ALLOW} decision or authorize reuse of an old grant.

\subsubsection{Evidence Binding}\label{evidence-binding}

The ERC records both \(Q_K\) and \(Q_P\), together with a commitment to the complete evidence set used for adjudication. Recording evidence without its obligations is insufficient because it does not show whether all required positive premises were considered.

Evidence binding establishes that identified evidence was evaluated under declared obligation and resolution rules. It does not establish the external truth of an assertion. At redemption, all time- and state-dependent evidence conditions must still hold. Otherwise, the grant is rejected or the candidate is re-adjudicated.

\subsubsection{Validity Interval}\label{validity-interval}

Each positive ERC contains a validity interval:

\[
[t_{start},t_{end}].
\]

This interval cannot extend any applicable validity condition imposed by the intent, root policy, operational policy, evidence obligations, or context. An empty effective interval cannot support grant issuance.

\subsubsection{Grant Lifecycle}\label{grant-lifecycle}

The baseline grant lifecycle is:

\begin{verbatim}
                  +----------> EXPIRED
                  |
ISSUED -----------+----------> REVOKED
  |
  +-- successful redemption -> CONSUMED
\end{verbatim}

\texttt{CONSUMED}, \texttt{EXPIRED}, and \texttt{REVOKED} are terminal for the grant instance. A second redemption attempt against any terminal state returns \texttt{DENY}.

Redemption, revocation, and expiry are competing transitions from \texttt{ISSUED}. Their authoritative state changes must be linearized. If revocation or expiry linearizes first, a concurrent Redemption observes a terminal state and is denied. If successful Redemption linearizes first, the grant becomes \texttt{CONSUMED} and a later revocation or expiry request cannot change that terminal state. This ordering requirement follows the standard linearizability criterion for concurrent objects~\cite{herlihy-wing1990}.

Grant issuance requires a verified \texttt{ALLOW} ERC and creates \(g\) such that:

\[
state(g)=ISSUED
\]

and

\[
GrantBoundTo(g,erc)=true.
\]

\subsubsection{Redemption Conditions}\label{redemption-conditions}

Let \(B_a\) be the original adjudication bundle, \(\pi\) its supplied derivation, and let

\[
R=\langle K_r,P_r,x_r,Ev_r,Ctx_r,t_r\rangle
\]

be the redemption-time inputs. Redemption is eligible only if:

\[
\begin{aligned}
&ValidRedeem(g,erc,B_a,\pi,R)\\
\iff{}&
VerifyERC(erc,B_a,\pi)=VALID\\
{}\land{}&erc.decision=ALLOW\\
{}\land{}&state(g)=ISSUED\\
{}\land{}&GrantBoundTo(g,erc)\\
{}\land{}&id(x_r)=erc.candidateId\\
{}\land{}&Scope(g)\subseteq erc.authorityScope\\
{}\land{}&t_r\in erc.interval\\
{}\land{}&CurrentConditionsHold(erc,B_a,R).
\end{aligned}
\]

In the baseline profile, \texttt{CurrentConditionsHold} is not an implementation-defined catch-all. It requires: exact Root and Operational Policy version equality; equality of the committed evidence and context identities; continued \(Bound(I,x_r,Ctx_r,t_r)\); continued satisfaction of both policy predicates; \texttt{VALID} status for every obligation in the committed \(Q_K\cup Q_P\); and satisfaction of every declared temporal predicate. Any extension predicate must be named, versioned, and included in the ERC commitment.

A successful redemption performs the logical transition:

\[
\begin{aligned}
\langle state(g)=ISSUED,s\rangle
&\xrightarrow{Redeem(g,x_r)}\\[-0.2em]
&\langle state(g)=CONSUMED,\Delta(s,x_r)\rangle.
\end{aligned}
\]

Validation, the \texttt{ISSUED}-to-\texttt{CONSUMED} transition, and the protected effect must be linearized with respect to decision-relevant state. Acceptable realizations may include an atomic transaction, version-conditional commit, transactional state transition, or another mechanism with equivalent semantics. A separate recheck followed by an interleavable effect does not satisfy this requirement.

EBL-Core specifies the linearization requirement, not the implementation mechanism.

\subsection{Threat Model and Non-Goals}\label{threat-model-and-non-goals}

The adversary may control the Agent and untrusted content available to it. The adversary may attempt to:

\begin{itemize}
\tightlist
\item
  construct an action inconsistent with the established intent;
\item
  substitute or mutate a candidate after adjudication;
\item
  omit execution-relevant parameters;
\item
  supply malformed, missing, expired, or conflicting evidence;
\item
  cause operational policy to omit root-policy obligations;
\item
  replay an earlier ERC or grant;
\item
  race two redemption attempts;
\item
  exploit a state transition between validation and effect;
\item
  present an ERC under different policy or profile versions; or
\item
  cause disagreement between an adjudicator and verifier.
\end{itemize}

Under the base model, operational policy cannot change \(Permit_K\), \(Q_K\), or the resolution semantics of root obligations. This restriction establishes root-policy dominance within the declared policy regime. It does not establish that the root policy is correct or that its administrative authority can never replace it.

The semantic properties of EBL-Core include:

\begin{itemize}
\tightlist
\item
  deterministic semantic decisions and reasons over equivalent closed inputs;
\item
  root-policy dominance under constrained obligation composition;
\item
  binding to one canonical candidate;
\item
  rejection when a positive evidence obligation is not discharged;
\item
  verification of a decision derivation against the complete input bundle;
\item
  single-use grant-consumption semantics; and
\item
  rejection when current redemption conditions do not match the ERC.
\end{itemize}

The following properties remain deployment assumptions:

\begin{itemize}
\tightlist
\item
  complete mediation of every relevant effect-producing path;
\item
  integrity and availability of trusted inputs;
\item
  protection of the applicable root-policy version;
\item
  correct implementation of the verifier, issuer, and Effector;
\item
  exclusive control of the authority needed for the protected effect;
\item
  linearizable redemption, consumption, and protected transition;
\item
  prevention of unauthorized credential extraction or delegation; and
\item
  faithful execution of the candidate by the Effector.
\end{itemize}

The model does not claim global non-bypassability, correct human-intent understanding, evidence truth, root-policy correctness, correct external outcomes, or security after compromise of every trusted role.

\begin{figure*}[t]
  \centering
  \includegraphics[width=\textwidth]{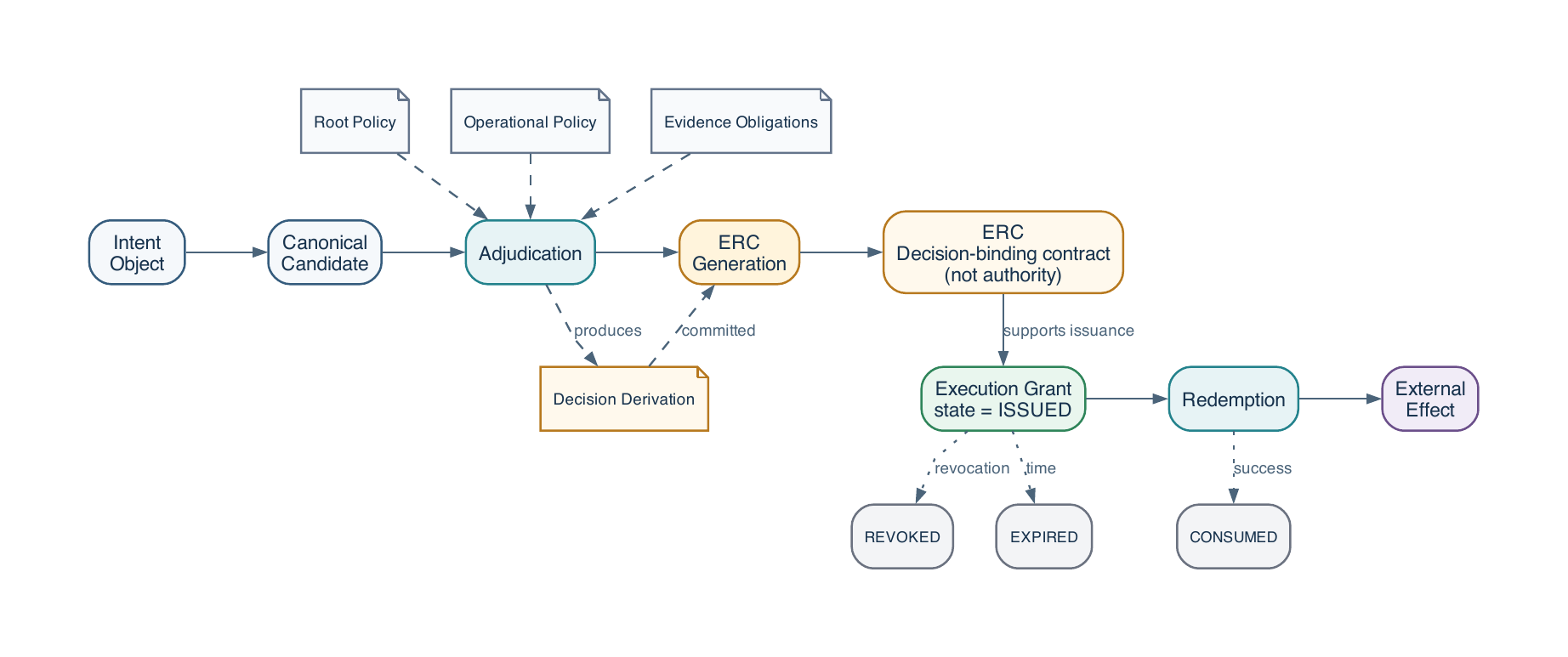}
  \caption{Execution-release lifecycle in EBL-Core. Adjudication consumes the Root Policy, Operational Policy, and Evidence Obligations and produces a Decision Derivation. ERC Generation creates a decision-binding contract, while separate Grant Issuance creates action-scoped authority. Redemption validates current conditions and linearizes the protected effect with the grant-state transition; \texttt{CONSUMED}, \texttt{EXPIRED}, and \texttt{REVOKED} are terminal grant states.}
  \label{fig:execution-release-lifecycle}
\end{figure*}

\section{Formal Execution-Boundary Semantics}\label{formal-execution-boundary-semantics}

\subsection{Semantic Domains and Canonical Identity}\label{semantic-domains-and-canonical-identity}

Let:

\[
B=
\langle
v,K,P,I,x,Ev,Ctx,t
\rangle
\]

be the complete input bundle for one adjudication, where \(v\) is the EBL-Core profile version.

For each decision-relevant object \(o\),

\[
id(o)=Commit_v(Canon_v(o)).
\]

The applicable profile defines canonical equality \(o_1\equiv_v o_2\) over all decision-relevant fields.

EBL-Core requires:

\begin{enumerate}
\def\labelenumi{\arabic{enumi}.}
\item
  \textbf{Canonical determinism}

  \[
  o_1\equiv_v o_2
  \implies
  Canon_v(o_1)=Canon_v(o_2).
  \]
\item
  \textbf{Computational binding}

  For distinct canonical decision-relevant objects, finding \(o_1\not\equiv_v o_2\) such that

  \[
  Commit_v(Canon_v(o_1))
  =
  Commit_v(Canon_v(o_2))
  \]

  is computationally infeasible under the selected commitment assumption.
\item
  \textbf{Stable replay identity}

  Implementations using the same profile and semantic object derive the same canonical representation and commitment.
\end{enumerate}

Canonical identity establishes equality within the model. It does not establish authenticity, authorization, provenance, or external truth.

\subsection{Intent-to-Candidate Binding Semantics}\label{intent-to-candidate-binding-semantics}

An intent object \(I\) is a trusted structured input, not a representation of unobservable mental intent. Let:

\[
Complete_{Core}(x)
\]

hold only when \(x\) denotes one candidate action and every execution-relevant parameter has a concrete committed value.

The EBL-Core binding predicate is:

{\small
\[
\begin{aligned}
Bound(I,x,Ctx,t)
\iff{}&
WellFormedIntent(I)\\
{}\land{}&AuthorizedIntentIssuer(I)\\
{}\land{}&t\in Window(I)\\
{}\land{}&Complete_{Core}(x)\\
{}\land{}&Scope(x)\sqsubseteq Scope(I)\\
{}\land{}&AuthorizedRefinement(I,x,Ctx,t)\\
{}\land{}&PreservesImmutableConstraints(I,x,Ctx,t).
\end{aligned}
\]
}

Permitted refinement may resolve an intent-authorized domain to one concrete operation, target, or parameter value. It cannot leave an execution-relevant choice to be selected after adjudication.

Consequently:

\[
x=\{x_1,\ldots,x_n\}
\quad\text{or}\quad
x=\text{an unresolved bounded action family}
\]

does not satisfy \(Complete_{Core}(x)\). Support for action families requires an extension profile with separate identity, refinement, and redemption semantics.

Intent binding establishes conformance between structured objects. It does not establish that \(I\) correctly captures a human request or that execution of \(x\) will achieve its declared purpose.

\subsection{Evidence-Obligation Semantics}\label{evidence-obligation-semantics}

Root and operational obligations are generated independently:

\[
Q_K=Obligations_K(K,I,x,Ctx),
\]

\[
Q_P=Obligations_P(P,I,x,Ctx),
\]

\[
Q=Q_K\cup Q_P.
\]

An obligation contains at least:

\begin{itemize}
\tightlist
\item
  a policy-origin identifier;
\item
  an assertion, subject, and scope;
\item
  admissible evidence types and sources;
\item
  freshness conditions;
\item
  conflict and aggregation rules; and
\item
  the versioned resolution semantics used to classify it.
\end{itemize}

For every \(q\in Q_K\), its identity and resolution semantics are determined by \(K\). Operational policy cannot remove, weaken, replace, or reinterpret it.

For each \(q\in Q\):

\[
\begin{aligned}
Status(q,Ev,Ctx,t)\in\{&VALID,UNKNOWN,MISSING,\\
&EXPIRED,CONFLICT\}.
\end{aligned}
\]

The status function is deterministic under the identified evidence profile. A profile must define mutually exclusive classification rules or a total precedence relation for overlapping diagnostic conditions. In particular, no obligation may be classified \texttt{VALID} while an applicable unresolved conflict remains.

The discharge predicate is:

\[
Discharged(q,Ev,Ctx,t)
\iff
Status(q,Ev,Ctx,t)=VALID.
\]

Define:

\[
\begin{aligned}
EvidenceOK_K(K,I,x,Ev,Ctx,t)&\iff{}\\
&\forall q\in Q_K:\ Discharged(q,Ev,Ctx,t),
\end{aligned}
\]

and:

\[
\begin{aligned}
EvidenceOK_P(P,I,x,Ev,Ctx,t)&\iff{}\\
&\forall q\in Q_P:\ Discharged(q,Ev,Ctx,t).
\end{aligned}
\]

The combined evidence condition is:

\[
EvidenceOK_{K,P}
=
EvidenceOK_K\land EvidenceOK_P.
\]

Therefore:

\[
Status(q,Ev,Ctx,t)\neq VALID
\implies
\neg Discharged(q,Ev,Ctx,t).
\]

These judgments establish satisfaction of declared obligations under the profile. They do not establish the external truth of the underlying assertion.

\subsection{Policy-Composition Semantics}\label{policy-composition-semantics}

Define:

{\small
\[
\begin{aligned}
RootOK_K(I,x,Ev,Ctx,t)
\iff{}&
Permit_K(I,x,Ev,Ctx,t)\\
{}\land{}&EvidenceOK_K(K,I,x,Ev,Ctx,t),
\end{aligned}
\]

and:

\[
\begin{aligned}
OperationalOK_P(I,x,Ev,Ctx,t)
\iff{}&
Permit_P(I,x,Ev,Ctx,t)\\
{}\land{}&EvidenceOK_P(P,I,x,Ev,Ctx,t).
\end{aligned}
\]
}

For the empty operational policy:

\[
Permit_{\varnothing}=true,
\qquad
Q_{\varnothing}=\varnothing,
\qquad
OperationalOK_{\varnothing}=true.
\]

For fixed \(\eta=\langle I,Ev,Ctx,t\rangle\), define:

\[
\begin{aligned}
Allow_{\eta}(K,P)=
\{x\mid{}&
WellFormed(B)\\
{}\land{}&Bound(I,x,Ctx,t)\\
{}\land{}&RootOK_K(I,x,Ev,Ctx,t)\\
{}\land{}&OperationalOK_P(I,x,Ev,Ctx,t)
\}.
\end{aligned}
\]

\subsubsection{Root-Policy Dominance}\label{root-policy-dominance}

Provided that \(P\) cannot modify \(Permit_K\), \(Q_K\), or the resolution semantics of root obligations:

\[
Allow_{\eta}(K,P)
\subseteq
Allow_{\eta}(K,\varnothing).
\]

This property is relative to one fixed root-policy version \(K\). It does not claim that \(K\) is correct, complete, or immutable across administrative changes.

\subsubsection{Conjunctive Rule Addition}\label{conjunctive-rule-addition}

Let \(P'=P\cup\{p\}\), where rule \(p\) is composed conjunctively:

\[
Permit_{P'}=Permit_P\land Permit_p
\]

and:

\[
Q_{P'}=Q_P\cup Q_p.
\]

Then:

\[
Allow_{\eta}(K,P')
\subseteq
Allow_{\eta}(K,P).
\]

This property applies only to rule addition under the specified conjunction and obligation-union semantics. It does not apply to arbitrary replacement, deletion, reprioritization, or reinterpretation of policy rules.

\subsubsection{Policy-Version Evolution}\label{policy-version-evolution}

Define a separate compatibility relation:

\[
Compat_K(P_{old},P_{new})
\]

only when an appropriate validation procedure establishes:

{\small
\[
\forall I,Ev,Ctx,t:
Allow_{\langle I,Ev,Ctx,t\rangle}(K,P_{new})
\subseteq
Allow_{\langle I,Ev,Ctx,t\rangle}(K,P_{old}).
\]
}

Version order, naming, or the presence of additional rules does not establish this relation.

\texttt{Compat\_K} classifies the policy change as non-expanding under fixed \(K\). It does not imply that every action previously allowed by \(P_{old}\) remains allowed by \(P_{new}\), and it does not automatically preserve grants issued under \(P_{old}\). Baseline ERC reuse requires exact policy-version equality.

A change to \(K\) is outside this operational-policy relation and requires a new adjudication.

\subsection{Adjudication Function}\label{adjudication-function}

The positive predicate is:

\[
\begin{aligned}
&CoreOK_K(P,I,x,Ev,Ctx,t)\\
\iff{}&WellFormed(B)\\
{}\land{}&Bound(I,x,Ctx,t)\\
{}\land{}&Permit_K(I,x,Ev,Ctx,t)\\
{}\land{}&Permit_P(I,x,Ev,Ctx,t)\\
{}\land{}&\forall q\in Q_K:
Status(q,Ev,Ctx,t)=VALID\\
{}\land{}&\forall q\in Q_P:
Status(q,Ev,Ctx,t)=VALID.
\end{aligned}
\]

The positive rule is:

\[
\frac{CoreOK_K(P,I,x,Ev,Ctx,t)}
{
\Gamma_K(P,I,x,Ev,Ctx,t)
=
(ALLOW,OK,\pi_{allow})
}
\tag{E-Allow}
\]

where:

\[
VerifyDerivation(\pi_{allow},B,ALLOW,OK)=true.
\]

If \texttt{CoreOK} does not hold:

\[
\frac{\neg CoreOK_K(P,I,x,Ev,Ctx,t)}
{
\Gamma_K(P,I,x,Ev,Ctx,t)
=
(DENY,FirstFailure(B),\pi_{deny})
}
\tag{E-Deny}
\]

with:

\[
VerifyDerivation(
\pi_{deny},
B,
DENY,
FirstFailure(B)
)=true.
\]

The profile defines a deterministic priority relation over simultaneous failures and commits that relation through \(profileVersion\). For EBL-Core, let the ordered failure sequence be:

{\small
\begin{equation}
\begin{aligned}
\mathcal{F}(B)=\langle{}
&(INPUT\_INVALID,\neg WellFormed(B)),\\
&(INTENT\_BINDING\_FAILED,\neg Bound(I,x,Ctx,t)),\\
&(ROOT\_POLICY\_DENY,\neg Permit_K(I,x,Ev,Ctx,t)),\\
&(ROOT\_EVIDENCE\_UNSATISFIED,\\
&\qquad\neg EvidenceOK_K(K,I,x,Ev,Ctx,t)),\\
&(OPERATIONAL\_POLICY\_DENY,\neg Permit_P(I,x,Ev,Ctx,t)),\\
&(OPERATIONAL\_EVIDENCE\_UNSATISFIED,\\
&\qquad\neg EvidenceOK_P(P,I,x,Ev,Ctx,t))
\rangle .
\end{aligned}
\label{eq:failure-precedence}
\end{equation}
}

\(FirstFailure(B)\) is the reason in the least-indexed pair whose predicate is true. Predicates after the first structurally undefined predicate are not evaluated; \texttt{INPUT\_INVALID} therefore dominates all semantic failures. Status selection within an individual obligation follows the separate total precedence declared by its evidence-resolution profile. Equation~\eqref{eq:failure-precedence} fixes primary-reason selection without collapsing the full set of diagnostic failures that a derivation may record.

ERC eligibility is:

{\small
\[
EligibleForERC(B,d,r,\pi)
\iff
VerifyDerivation(\pi,B,d,r)=true.
\]
}

\subsubsection{Determinism}\label{determinism}

For canonically equivalent bundles evaluated under the same semantic profile, conforming implementations must produce the same semantic decision and primary reason:

\[
Canon_v(B_1)=Canon_v(B_2)
\implies
d_1=d_2
\land
r_1=r_2.
\]

Each supplied derivation must verify:

\[
VerifyDerivation(\pi_i,B_i,d_i,r_i)=true.
\]

Different conforming implementations may use different valid derivation encodings. EBL-Core does not require:

\[
id(\pi_1)=id(\pi_2).
\]

An individual ERC nevertheless commits to the particular derivation supplied with that ERC.

\subsubsection{Side-Effect Freedom, Input Closure, and Bounded Evaluation}\label{side-effect-freedom-input-closure-and-bounded-evaluation}

Adjudication does not modify external state, policy state, evidence sources, grant state, or protected system state.

Its semantic result depends only on \(B\). A conforming evaluator does not consult an undeclared clock, network source, mutable service, randomness source, or hidden model inference.

Each profile specifies an input fragment and resource-bound function:

\[
Steps(\Gamma,B)\leq Bound_v(|B|).
\]

This is a conformance requirement, not a claim that a particular implementation has already been measured or verified against the bound.

\subsection{Decision-Derivation Semantics}\label{decision-derivation-semantics}

A decision derivation is a structured witness for one adjudication:

\[
\begin{aligned}
\pi=\langle{}&encodingVersion,profileVersion,\\
&inputCommitments,policyCommitments,\\
&rootObligations,operationalObligations,\\
&evidenceStatuses,inferenceSteps,\\
&decision,reason\rangle.
\end{aligned}
\]

Define:

\[
VerifyDerivation(\pi,B,d,r)=true
\]

only if:

\begin{enumerate}
\def\labelenumi{\arabic{enumi}.}
\tightlist
\item
  the derivation and \(B\) identify the same supported profile version;
\item
  the derivation commits to every component of \(B\);
\item
  the recorded root and operational policy versions match \(B\);
\item
  the recorded obligation sets equal the obligation sets generated from \(B\);
\item
  every evidence status follows from the committed evidence, context, time, and obligation semantics;
\item
  every derivation leaf corresponds to a committed input;
\item
  every inference step is admitted by the profile;
\item
  no required positive or negative premise has been omitted;
\item
  the terminal judgment is \(d\) with primary reason \(r\); and
\item
  the supplied derivation conforms to its declared encoding version.
\end{enumerate}

\texttt{VerifyDerivation} is profile-relative. It is not a universal proof system.

An ERC additionally checks:

\[
id(\pi)=erc.derivationId.
\]

This equality binds the ERC to the particular supplied derivation. It does not require an independent evaluator to serialize its own derivation identically.

Let:

\[
Replay_m(B)=(d_m,r_m,\pi_m)
\]

be replay by conforming implementation \(m\). Replay is semantically consistent with an ERC when:

\[
d_m=erc.decision,
\qquad
r_m=erc.reason,
\]

and:

\[
VerifyDerivation(\pi_m,B,d_m,r_m)=true.
\]

Different valid derivation encodings may therefore support the same replay result.

\subsection{Invalidation, Grant State, and Atomic Redemption}\label{invalidation-grant-state-and-atomic-redemption}

Define the adjudication fingerprint:

\[
\begin{aligned}
F(B)=Commit(&
profileVersion,id(I),id(x),\\
&version(K),version(P),id(Q_K),id(Q_P),\\
&id(Ev),id(Ctx),t).
\end{aligned}
\]

A baseline ERC becomes non-reusable if:

\begin{itemize}
\tightlist
\item
  the candidate identity changes;
\item
  the root- or operational-policy version changes;
\item
  the committed evidence set changes;
\item
  a required evidence status is no longer \texttt{VALID};
\item
  decision-relevant context changes;
\item
  intent binding no longer holds;
\item
  the validity interval expires; or
\item
  the supplied derivation no longer verifies against the original bundle.
\end{itemize}

Invalidation means that the earlier result cannot be reused. It does not convert the earlier \texttt{ALLOW} into a current \texttt{DENY}; a current decision requires re-adjudication.

Let the redemption-time input be

\[
R=\langle K_r,P_r,x_r,Ev_r,Ctx_r,t_r\rangle.
\]

For the baseline profile, current-condition validity is the following explicit conjunction:

\begin{equation}
\begin{aligned}
&CurrentConditionsHold(erc,B,R)\iff{}\\
&\quad version(K_r)=erc.rootPolicyVersion\\
&{}\land version(P_r)=erc.operationalPolicyVersion\\
&{}\land id(Ev_r)=erc.evidenceId\\
&{}\land id(Ctx_r)=erc.contextId\\
&{}\land Bound(I,x_r,Ctx_r,t_r)\\
&{}\land Permit_{K_r}(I,x_r,Ev_r,Ctx_r,t_r)\\
&{}\land Permit_{P_r}(I,x_r,Ev_r,Ctx_r,t_r)\\
&{}\land \forall q\in Q_K:\ Status(q,Ev_r,Ctx_r,t_r)=VALID\\
&{}\land \forall q\in Q_P:\ Status(q,Ev_r,Ctx_r,t_r)=VALID.
\end{aligned}
\label{eq:current-conditions}
\end{equation}

The obligation identities and resolution versions are already committed by the verified ERC. An extension may add current-state predicates only by naming and versioning them in the profile and committing their inputs in the ERC; undeclared resolver state is inadmissible.

The grant-state transition relation is:

\[
ISSUED\rightarrow
\{CONSUMED,EXPIRED,REVOKED\}.
\]

No transition returns a terminal grant to \texttt{ISSUED}.

A nonce contributes to ERC or grant identity. It does not establish this state transition. At-most-once redemption requires an authoritative grant-state transition or a mechanism with equivalent semantics.

Revocation and expiry are explicit competing transitions:

\[
\begin{aligned}
g:ISSUED&\xrightarrow{Revoke}g:REVOKED,\\
g:ISSUED&\xrightarrow{Expire}g:EXPIRED.
\end{aligned}
\]

All transitions out of \texttt{ISSUED} share one authoritative linearization order. Thus, if \texttt{Revoke} or \texttt{Expire} linearizes before Redemption, \texttt{R-Success} is disabled; if \texttt{R-Success} linearizes first, the later lifecycle operation cannot replace \texttt{CONSUMED}. This supplies a determinate result even when requests overlap in real time~\cite{herlihy-wing1990}.

The Redemption guard is defined, rather than left implementation-specific, by:

\begin{equation}
\begin{aligned}
&RedeemOK(g,erc,B,\pi,R)\iff{}\\
&\quad VerifyERC(erc,B,\pi)=VALID\\
&{}\land erc.decision=ALLOW\\
&{}\land state(g)=ISSUED\\
&{}\land GrantBoundTo(g,erc)\\
&{}\land id(x_r)=erc.candidateId\\
&{}\land Scope(g)\subseteq erc.authorityScope\\
&{}\land t_r\in erc.interval\\
&{}\land CurrentConditionsHold(erc,B,R).
\end{aligned}
\label{eq:redeem-ok}
\end{equation}

A successful protected transition has the form:

\[
\frac{
RedeemOK(g,erc,B,\pi,R)
\quad
state(g)=ISSUED
}{
\langle g:ISSUED,s\rangle
\xrightarrow{Redeem(g,x_r)}
\langle g:CONSUMED,\Delta(s,x_r)\rangle
}.
\tag{R-Success}
\]

The validation predicates, grant consumption, and protected effect must share one logical linearization point with respect to decision-relevant state. A conforming realization may use an atomic transaction, version-conditional commit, transactional state transition, or equivalent mechanism.

A validation step followed by an interleavable state change and later effect is not a conforming realization of \texttt{R-Success}.

If validation fails, if the grant is not \texttt{ISSUED}, or if linearization cannot be established:

\[
Redeem(g,\ldots)=DENY.
\]

\subsection{Conditional Security Propositions}\label{conditional-security-properties}

The following propositions are preservation results over the EBL-Core rules. Several are direct consequences of explicit guards; this is intentional because the profile is meant to make those guards testable. Each result states the assumptions on which it depends and the mutation that would falsify the conclusion if the corresponding guard were omitted.

\begin{proposition}[Exact action binding]\label{property-1-exact-action-binding}
Assume collision resistance for \(Commit_v\), a verified ERC, and a successful application of \texttt{R-Success} for candidate \(x_r\). Then
\[
id(x_r)=erc.candidateId.
\]
Consequently, the declared Redemption Interface cannot use a grant issued for \(x\) to redeem a decision-relevantly different candidate \(x'\).
\end{proposition}

\begin{proof}[Proof sketch]
Equation~\eqref{eq:redeem-ok} is a premise of \texttt{R-Success} and contains the equality \(id(x_r)=erc.candidateId\). A mutation of any canonical decision-relevant candidate field changes \(Canon_v(x_r)\) and hence its commitment, except with a collision excluded by the assumption. Removing the candidate-identity guard admits the recipient-substitution counterexample used in the conformance corpus. The result is scoped to the declared interface and says nothing about alternative effect-producing paths.
\end{proof}

\begin{proposition}[Root-policy dominance]\label{property-2-root-policy-dominance}
Fix a root-policy version \(K\). If Operational Policy cannot modify \(Permit_K\), \(Q_K\), or root-obligation resolution semantics, and if direct predicates and obligations compose conjunctively, then
\[
Allow_{\eta}(K,P)
\subseteq
Allow_{\eta}(K,\varnothing).
\]
\end{proposition}

\begin{proof}[Proof sketch]
Membership in \(Allow_{\eta}(K,P)\) requires both \(RootOK_K\) and \(OperationalOK_P\). The empty Operational Policy replaces only the latter conjunct with \(true\) and removes only \(Q_P\); it leaves \(RootOK_K\) unchanged. Every member of the first set is therefore a member of the second. Permitting \(P\) to delete or reinterpret a root obligation would invalidate the argument, which is why the non-weakening premise is explicit. The proposition does not establish that \(K\) is substantively correct or prevent an authorized replacement of \(K\).
\end{proof}

\begin{proposition}[Evidence-obligation safety]\label{property-3-evidence-obligation-safety}
If
\[
\Gamma_K(P,I,x,Ev,Ctx,t)=(ALLOW,r,\pi),
\]
then
\[
\forall q\in Q_K\cup Q_P:\ Status(q,Ev,Ctx,t)=VALID.
\]
\end{proposition}

\begin{proof}[Proof sketch]
The sole positive rule, \texttt{E-Allow}, requires \(CoreOK_K\). Its final two conjuncts require \texttt{VALID} for every member of \(Q_K\) and \(Q_P\). Replacing any required status by \texttt{UNKNOWN}, \texttt{MISSING}, \texttt{EXPIRED}, or \texttt{CONFLICT} falsifies \(CoreOK_K\) and selects \texttt{E-Deny} under the precedence in Equation~\eqref{eq:failure-precedence}. The result concerns obligation resolution, not the external truth of an evidence assertion.
\end{proof}

\begin{proposition}[Semantic replay consistency]\label{property-4-semantic-replay-consistency}
Assume input closure, deterministic canonicalization, deterministic evidence classification, and the fixed failure precedence in Equation~\eqref{eq:failure-precedence}. For closed bundles under the same profile,
\[
Canon_v(B_1)=Canon_v(B_2)
\implies
d_1=d_2\land r_1=r_2.
\]
Each accepted derivation must independently satisfy
\[
VerifyDerivation(\pi_i,B_i,d_i,r_i)=true.
\]
Serialized derivation equality is not required.
\end{proposition}

\begin{proof}[Proof sketch]
Canonical equality gives identical values for every declared input to the predicates and status functions. Input closure excludes an undeclared clock, resolver, randomness source, or mutable process value. The same predicates therefore fail or succeed, and the fixed total precedence selects the same primary reason. Different encodings may record different valid inference structures, so the conclusion is semantic rather than byte-level equality. A hidden clock or implementation-dependent iteration order provides a counterexample if either assumption is removed.
\end{proof}

\begin{proposition}[Single-use grant consumption]\label{property-5-single-use-grant-consumption}
Assume one authoritative, linearizable grant-state object and terminal states \texttt{CONSUMED}, \texttt{EXPIRED}, and \texttt{REVOKED}. At most one concurrent Redemption of the same grant can apply \texttt{R-Success} and produce the protected effect.
\end{proposition}

\begin{proof}[Proof sketch]
Linearizability totally orders competing state transitions from \texttt{ISSUED}~\cite{herlihy-wing1990}. The first successful Redemption changes the state to \texttt{CONSUMED}; every other Redemption is ordered after that transition and fails the \(state(g)=ISSUED\) premise. If revocation or expiry is first, no Redemption succeeds. A nonce alone would not prove the result because two consumers could accept the same nonce without a shared authoritative transition.
\end{proof}

\begin{proposition}[Conditional interface enforcement]\label{property-6-conditional-interface-enforcement}
Let \(ProtectedExecute(x)\) denote a protected transition through the declared interface. Assume: every such transition requires successful Redemption; the Effector applies only the redeemed candidate; Grant Issuance requires a verified \texttt{ALLOW} ERC; \texttt{R-Success} checks Equation~\eqref{eq:redeem-ok}; validation, consumption, and effect are linearized; grants cannot be forged or broadened within the model; and participating roles satisfy their declared trust assumptions. Then
\[
\begin{aligned}
ProtectedExecute(x)\implies{}&\exists B,\pi,erc,g:\\
&erc.decision=ALLOW,\\
&id(x)=erc.candidateId,\\
&VerifyDerivation(\pi,B,ALLOW,\\[-0.2em]
&\qquad erc.reason)=true,\\
&state(g):ISSUED\rightarrow CONSUMED.
\end{aligned}
\]
\end{proposition}

\begin{proof}[Proof sketch]
By the complete-mediation assumption for the declared interface, \(ProtectedExecute(x)\) has a successful Redemption witness. The only successful rule is \texttt{R-Success}; expanding its \(RedeemOK\) premise with Equation~\eqref{eq:redeem-ok} yields the ERC decision, candidate equality, verified derivation, and state transition in the conclusion. Removing complete mediation admits an alternative-path counterexample, while removing Effector fidelity permits execution of a candidate other than the committed one. Hence the proposition is interface-local and conditional; it does not establish deployment-wide non-bypassability, human-intent correctness, evidence truth, root-policy correctness, or correct external outcomes.
\end{proof}

\section{EBL-Core Conformance Profile and Integration Model}\label{ebl-core-conformance-profile-and-integration-model}

\subsection{Conformance Model Overview}\label{conformance-model-overview}

An EBL-Core implementation assigns the following logical roles:

\begin{enumerate}
\def\labelenumi{\arabic{enumi}.}
\tightlist
\item
  Intent Provider;
\item
  Candidate Materializer;
\item
  Evidence Resolver;
\item
  Adjudicator;
\item
  ERC Generator;
\item
  ERC Verifier;
\item
  Grant Issuer; and
\item
  Effector and Redemption Interface.
\end{enumerate}

The principal path is:

\begin{verbatim}
Candidate Materialization
        |
Adjudication
        |
ERC Generation and Verification
        |
Grant Issuance: state = ISSUED
        |
Linearized Redemption and Consumption
        |
Protected Effect or Denial
\end{verbatim}

The Candidate Materializer produces one canonical action containing every field required by the applicable candidate profile. A schema cannot guarantee coverage of unknown real-world effects; candidate-profile completeness and Effector fidelity remain explicit assumptions.

The Adjudicator produces a deterministic semantic decision and reason together with a verifiable decision derivation. Different implementations may encode valid derivations differently.

The ERC Generator produces a decision-binding release-condition object. ERC generation does not release execution authority.

The ERC Verifier checks the ERC against the original adjudication bundle and supplied derivation. Current-state redemption compatibility is evaluated separately by the Redemption Interface.

The Grant Issuer releases authority for the exact candidate and initializes the grant in the \texttt{ISSUED} state.

The Redemption Interface verifies current conditions and performs a linearized transition that consumes the grant and releases the protected effect. Whether all effect-producing paths traverse this interface remains a deployment property.

\subsection{EBL-Core Input Contract}\label{ebl-core-input-contract}

The Adjudicator consumes:

\[
B=
\langle
profileVersion,K,P,I,x,Ev,Ctx,t
\rangle.
\]

A conforming candidate \(x\) denotes exactly one action. It must be fully materialized, schema-valid, deterministically canonicalizable, and complete with respect to all decision-relevant parameters.

The policy inputs separately identify:

\[
Q_K=Obligations_K(K,I,x,Ctx)
\]

and:

\[
Q_P=Obligations_P(P,I,x,Ctx).
\]

The adapter must preserve the origin and semantics of each obligation. Operational-policy processing cannot suppress or reinterpret root-policy obligations.

The evidence set \(Ev\) must contain the complete materialized evidence input used to classify \(Q_K\cup Q_P\). Undeclared resolver state cannot affect adjudication.

The context \(Ctx\) must contain the complete decision-relevant contextual projection. Time \(t\) remains an explicit input.

Input closure requires the result to depend only on \(B\) and the referenced profile semantics. External observations may occur before adjudication, but their decision-relevant results must be materialized in \(Ev\) or \(Ctx\).

\subsection{Adjudicator Conformance Requirements}\label{adjudicator-conformance-requirements}

\subsubsection{Deterministic Semantic Evaluation}\label{deterministic-semantic-evaluation}

For canonically equivalent bundles:

\[
Canon_v(B_1)=Canon_v(B_2),
\]

conforming evaluators must produce:

\[
d_1=d_2
\qquad\text{and}\qquad
r_1=r_2.
\]

Each derivation must satisfy:

\[
VerifyDerivation(\pi_i,B_i,d_i,r_i)=true.
\]

Conformance does not require:

\[
id(\pi_1)=id(\pi_2).
\]

The profile must nevertheless define deterministic semantics for rule resolution, failure priority, evidence-status classification, and unordered inputs. An implementation-specific iteration order cannot alter \(d\) or \(r\).

\subsubsection{Side-Effect-Free and Bounded Evaluation}\label{side-effect-free-and-bounded-evaluation}

Adjudication must not modify external state, policies, evidence sources, grant state, redemption state, or the protected system.

Each profile defines:

\[
Steps(\Gamma,B)\leq Bound_{Profile}(|B|).
\]

The present specification defines this obligation without claiming that a particular implementation has already been measured or verified against it.

\subsubsection{Input Validation and Failure Classification}\label{input-validation-and-failure-classification}

Before positive adjudication, the Adjudicator validates:

\begin{itemize}
\tightlist
\item
  profile and schema versions;
\item
  canonical representability;
\item
  intent and candidate structure;
\item
  exact candidate completeness;
\item
  policy identities;
\item
  separate root and operational obligation sets; and
\item
  evidence-obligation structure.
\end{itemize}

Malformed, incomplete, unsupported, or unresolved evaluation cannot produce \texttt{ALLOW}.

\subsubsection{No Implicit Repair}\label{no-implicit-repair}

Obtaining new evidence, modifying a candidate, replacing a policy, resolving a conflict, or filling a missing parameter creates a new input bundle and requires a new adjudication.

\subsection{ERC Generation, Grant Issuance, and Redemption}\label{erc-generation-grant-issuance-and-redemption}

Let:

\[
R=
\langle
K_r,P_r,x_r,Ev_r,Ctx_r,t_r
\rangle
\]

denote the current redemption inputs.

The abstract interfaces are:

\[
Adjudicate(B)\rightarrow(d,r,\pi),
\]

\[
GenerateERC(B,d,r,\pi,issuer,nonce,interval)
\rightarrow erc,
\]

\[
VerifyERC(erc,B,\pi)
\rightarrow
\{VALID,INVALID\},
\]

\[
IssueGrant(erc,B,\pi)
\rightarrow
g\mid\bot,
\]

and:

\[
Redeem(g,erc,B,\pi,R)
\rightarrow
\{EXECUTE,DENY\}.
\]

\subsubsection{ERC Generation}\label{erc-generation}

The ERC Generator must bind:

\begin{itemize}
\tightlist
\item
  the complete input-bundle commitment;
\item
  intent and exact-candidate identities;
\item
  actor and authority scope;
\item
  root- and operational-policy versions;
\item
  \(Q_K\) and \(Q_P\);
\item
  evidence and context commitments;
\item
  adjudication time and validity interval;
\item
  decision and reason;
\item
  the supplied derivation commitment;
\item
  profile and schema versions;
\item
  issuer identity; and
\item
  a nonce.
\end{itemize}

For an \texttt{ALLOW} ERC:

\[
VerifyDerivation(\pi,B,ALLOW,r)=true.
\]

A \texttt{DENY} ERC may preserve a refusal for verification or diagnosis, but:

\[
erc.decision=DENY
\implies
IssueGrant(erc,B,\pi)=\bot.
\]

The nonce distinguishes the ERC instance. It does not implement grant consumption.

\subsubsection{ERC Verification}\label{erc-verification}

\(VerifyERC(erc,B,\pi)=VALID\) only if:

\begin{itemize}
\tightlist
\item
  the ERC is well formed;
\item
  profile and schema versions are recognized;
\item
  all object commitments match \(B\);
\item
  \(Q_K\) and \(Q_P\) are correctly generated and separately identified;
\item
  the candidate denotes exactly one complete action;
\item
  the issuer is authorized for the recorded scope;
\item
  the supplied derivation commitment matches the ERC's \texttt{derivationId} field; and
\item
  the supplied derivation verifies for \(B\), \(erc.decision\), and \(erc.reason\).
\end{itemize}

ERC verification establishes consistency with the original adjudication. It does not establish that current redemption conditions remain valid.

\subsubsection{Grant Issuance}\label{grant-issuance}

A grant may be issued only if:

\[
erc.decision=ALLOW
\land
VerifyERC(erc,B,\pi)=VALID.
\]

Issuance creates \(g\) such that:

\[
GrantBoundTo(g,erc)=true,
\]

\[
CandidateOf(g)=erc.candidateId,
\]

\[
Scope(g)\subseteq erc.authorityScope,
\]

and:

\[
state(g)=ISSUED.
\]

The representation of \(g\) is an implementation choice. An ERC may be embedded in the same credential that represents \(g\), but the contract and authority-bearing roles remain semantically distinct.

\subsubsection{Redemption Verification}\label{redemption-verification}

Before the protected effect, the Redemption Interface verifies:

\[
\begin{aligned}
&VerifyERC(erc,B,\pi)=VALID,\\
&GrantBoundTo(g,erc),\\
&state(g)=ISSUED,\\
&id(x_r)=erc.candidateId,\\
&version(K_r)=erc.rootPolicyVersion,\\
&version(P_r)=erc.operationalPolicyVersion,\\
&t_r\in erc.validityInterval,\\
&CurrentConditionsHold(erc,B,R).
\end{aligned}
\]

For the baseline profile, \texttt{CurrentConditionsHold} expands to the conjunction in Equation~\eqref{eq:current-conditions}: exact policy-version equality, equality of committed evidence and context identities, continued intent binding, both policy predicates, and \texttt{VALID} status for every committed obligation. An adapter cannot place an undeclared network lookup, clock, resolver, or policy decision inside this predicate. Extension conditions must be named and versioned by the profile and bound by the ERC.

If any condition is false or unresolved, redemption returns \texttt{DENY}.

For a successful redemption, validation, the transition

\[
ISSUED\rightarrow CONSUMED,
\]

and the protected effect must be linearized. A conforming implementation may realize this requirement through an atomic transaction, version-conditional commit, transactional state transition, or an equivalent mechanism. EBL-Core does not prescribe which mechanism is used.

A grant in \texttt{CONSUMED}, \texttt{EXPIRED}, or \texttt{REVOKED} cannot be redeemed. The nonce assists with identity and correlation, but the authoritative grant-state transition supplies replay prevention.

Redemption, revocation, and expiry must operate on the same authoritative lifecycle state. When they overlap, the first transition linearized from \texttt{ISSUED} determines the terminal state: revocation or expiry first forces Redemption to deny; Redemption first produces \texttt{CONSUMED}, after which revocation or expiry cannot overwrite the result.

The EBL-Core baseline is single-candidate and single-use. Bounded candidate sets and multi-use grants require an extension profile.

\subsection{Worked Example: A Single-Use Transfer Grant}\label{worked-transfer-example}

Consider an authenticated treasury request to transfer at most 10,000 USDC from account \texttt{treasury-01} to the pre-approved recipient \texttt{beneficiary-alice} on \texttt{mainnet}. The Intent Authority emits a structured intent \(I_T\) valid over \([1700000000,1700001000]\). The Candidate Materializer resolves the request to exactly one candidate \(x_T\): transfer 7,500 USDC to that recipient, with a maximum fee of 25 units and authority scope \texttt{transfer:treasury-01}. No recipient, amount, network, or fee remains unresolved.

\begin{table*}[t]
\centering
\small
\setlength{\tabcolsep}{4pt}
\renewcommand{\arraystretch}{1.12}
\begin{tabular}{@{}>{\raggedright\arraybackslash}p{0.19\textwidth}>{\raggedright\arraybackslash}p{0.34\textwidth}>{\raggedright\arraybackslash}p{0.42\textwidth}@{}}
\toprule
\textbf{Object} & \textbf{Materialized value} & \textbf{Decision-relevant role} \\
\midrule
Intent \(I_T\) & Named actor, source account, asset, recipient set, 10,000-unit maximum, network, scope, and validity window & Defines immutable constraints and permitted refinement for one transfer candidate. \\
Candidate \(x_T\) & 7,500 USDC; \texttt{beneficiary-alice}; \texttt{mainnet}; maximum fee 25 & Supplies the exact payload committed by adjudication, ERC Generation, Grant Issuance, and Redemption. \\
Root Policy \(K_3\) & Amount at most 10,000; USDC and mainnet permitted; account enabled and funded & Generates \(q_1=\) recipient allowlisting and \(q_2=\) sufficient-balance obligations. \\
Operational Policy \(P_{17}\) & Amount at most 9,000 & Generates \(q_3=\) approval-quorum obligation without changing \(q_1\) or \(q_2\). \\
Evidence \(Ev_T\) & Allowlist attestation, ledger snapshot, and approval record, each bound to \(x_T\) and current at \(t=1700000100\) & Resolves \(Status(q_i,Ev_T,Ctx_T,t)=VALID\) for \(i\in\{1,2,3\}\). \\
Context \(Ctx_T\) & Account enabled, balance 10,000, state version \texttt{ledger-state-42} & Supplies the explicit state projection used by both policy predicates. \\
\bottomrule
\end{tabular}
\caption{Complete materialized inputs for the financial-transfer trace. Values are illustrative units and identifiers; they do not describe a production payment deployment.}
\label{tab:worked-transfer-inputs}
\end{table*}

The bundle is well formed, intent binding and both policy predicates hold, and every evidence obligation resolves to \texttt{VALID}. Hence no predicate in the EBL-Core failure sequence is true, and adjudication produces

\[
\Gamma_{K_3}(P_{17},I_T,x_T,Ev_T,Ctx_T,t)
=(ALLOW,OK,\pi_T).
\]

The supplied derivation \(\pi_T\) records the input commitments, the separate sets \(Q_K=\{q_1,q_2\}\) and \(Q_P=\{q_3\}\), the three \texttt{VALID} classifications, and the applied \texttt{E-Allow} rule. ERC Generation then creates \(erc_T\) containing the commitments to \(I_T\), \(x_T\), \(K_3\), \(P_{17}\), \(Q_K\), \(Q_P\), \(Ev_T\), \(Ctx_T\), \(t\), \(\pi_T\), the effective interval, issuer, and nonce. Verification of \(erc_T\) does not itself release authority. Separate Grant Issuance creates \(g_T\) in \texttt{ISSUED}, bound to \(erc_T\) and the candidate commitment \(id(x_T)\).

At Redemption, the interface receives the same candidate, policy versions, evidence and context commitments, and a current time inside the effective interval. Equations~\eqref{eq:current-conditions} and~\eqref{eq:redeem-ok} hold. A successful linearized transition applies the in-scope effect once and changes \(g_T\) from \texttt{ISSUED} to \texttt{CONSUMED}.

Three perturbations expose the bindings. First, changing the recipient to \texttt{beneficiary-bob} changes \(id(x_T)\), so the existing ERC and grant fail exact-candidate validation. Second, redeeming after the approval evidence or effective ERC interval expires makes the relevant obligation non-\texttt{VALID} and disables \texttt{R-Success}. Third, if 32 Redemption requests race on \(g_T\), the authoritative state object orders them: at most one can observe and consume \texttt{ISSUED}; the rest observe \texttt{CONSUMED} and deny. Section~\ref{preliminary-reference-artifact} reports executable checks of these cases.

\subsection{Interoperability with Existing Systems}\label{interoperability-with-existing-systems}

\begin{table*}[t]
\centering
\small
\setlength{\tabcolsep}{4pt}
\renewcommand{\arraystretch}{1.12}
\begin{tabular}{@{}>{\raggedright\arraybackslash}p{0.18\textwidth}>{\raggedright\arraybackslash}p{0.38\textwidth}>{\raggedright\arraybackslash}p{0.39\textwidth}@{}}
\toprule
\textbf{Existing mechanism} & \textbf{Possible EBL-Core role} & \textbf{Adapter obligation} \\
\midrule
Cedar, Rego, or XACML & Root or operational policy evaluation & Preserve separate root and operational predicates and obligations; expose exact policies, inputs, and results. \\
Capability or credential system & Execution-grant representation & Bind authority to the exact candidate and ERC, expose grant state, and prevent broader interpretation. \\
Provenance or evidence system & Evidence production & Identify assertion, source, scope, freshness, and the root or operational obligation to which the item applies. \\
Agent runtime & Candidate materialization & Produce one complete canonical candidate with no post-adjudication parameter resolution. \\
Reference monitor or enforcement point & Redemption boundary & Verify current conditions and linearize validation, single-use consumption, and the protected transition. \\
\bottomrule
\end{tabular}
\end{table*}

A capability system may represent \(g\), but EBL-Core does not redefine delegation or credential transport. If the capability cannot bind the candidate identity, grant state, or redemption conditions, an additional conforming verifier is required.

Cross-system replay requires the same semantic decision and reason under equivalent input bundles. Independent implementations may emit different valid derivation encodings, provided that each satisfies:

\[
VerifyDerivation(\pi,B,d,r)=true.
\]

Interoperability therefore concerns semantic equivalence and verification at the ERC interface, not byte-identical derivation output or shared internal code.

\subsection{Conformance Levels}\label{conformance-levels}

\subsubsection{Level 0: Decision-Compatible}\label{level-0-decision-compatible}

A Level 0 implementation provides an \texttt{ALLOW} or \texttt{DENY} decision, stable primary reason, identified profile version, and fail-closed handling of malformed or unavailable evaluation. It does not provide full execution-release conformance.

\subsubsection{Level 1: Action-Bound Adjudication}\label{level-1-action-bound-adjudication}

Level 1 adds:

\begin{itemize}
\tightlist
\item
  one canonical candidate identity;
\item
  a canonical intent identity;
\item
  explicit intent-to-candidate binding;
\item
  separate root and operational obligations;
\item
  committed policy and input versions;
\item
  deterministic semantic decisions and reasons; and
\item
  a decision derivation satisfying the profile verification relation.
\end{itemize}

Level 1 does not establish that execution authority is bound to the decision.

\subsubsection{Level 2: Release-Contract Conformance}\label{level-2-release-contract-conformance}

Level 2 adds:

\begin{itemize}
\tightlist
\item
  generation and verification of a complete ERC;
\item
  exact-candidate grant binding;
\item
  initialization of grant state as \texttt{ISSUED};
\item
  validity and current-condition checks;
\item
  single-use \texttt{ISSUED}-to-\texttt{CONSUMED} semantics;
\item
  rejection of terminal grant states; and
\item
  linearized validation, consumption, and protected transition at the declared interface.
\end{itemize}

Level 2 is the minimum level for claiming EBL-Core execution-release conformance. The claim remains scoped to the declared interface and does not establish the absence of alternative effect-producing paths.

\subsubsection{Level 3: Independently Replayable Boundary}\label{level-3-independently-replayable-boundary}

Level 3 adds:

\begin{itemize}
\tightlist
\item
  an independently specified ERC verifier;
\item
  complete decision-derivation verification or semantic replay;
\item
  a published conformance corpus;
\item
  negative, mutation, policy-composition, and concurrent-redemption tests;
\item
  cross-runtime or cross-implementation evaluation; and
\item
  documented equality of semantic decisions and reasons.
\end{itemize}

Level 3 does not require byte-identical derivation serialization. It also does not establish evidence truth, correct human intent, root-policy correctness, or deployment-wide non-bypassability.

\begin{figure*}[t]
  \centering
  \includegraphics[width=\textwidth]{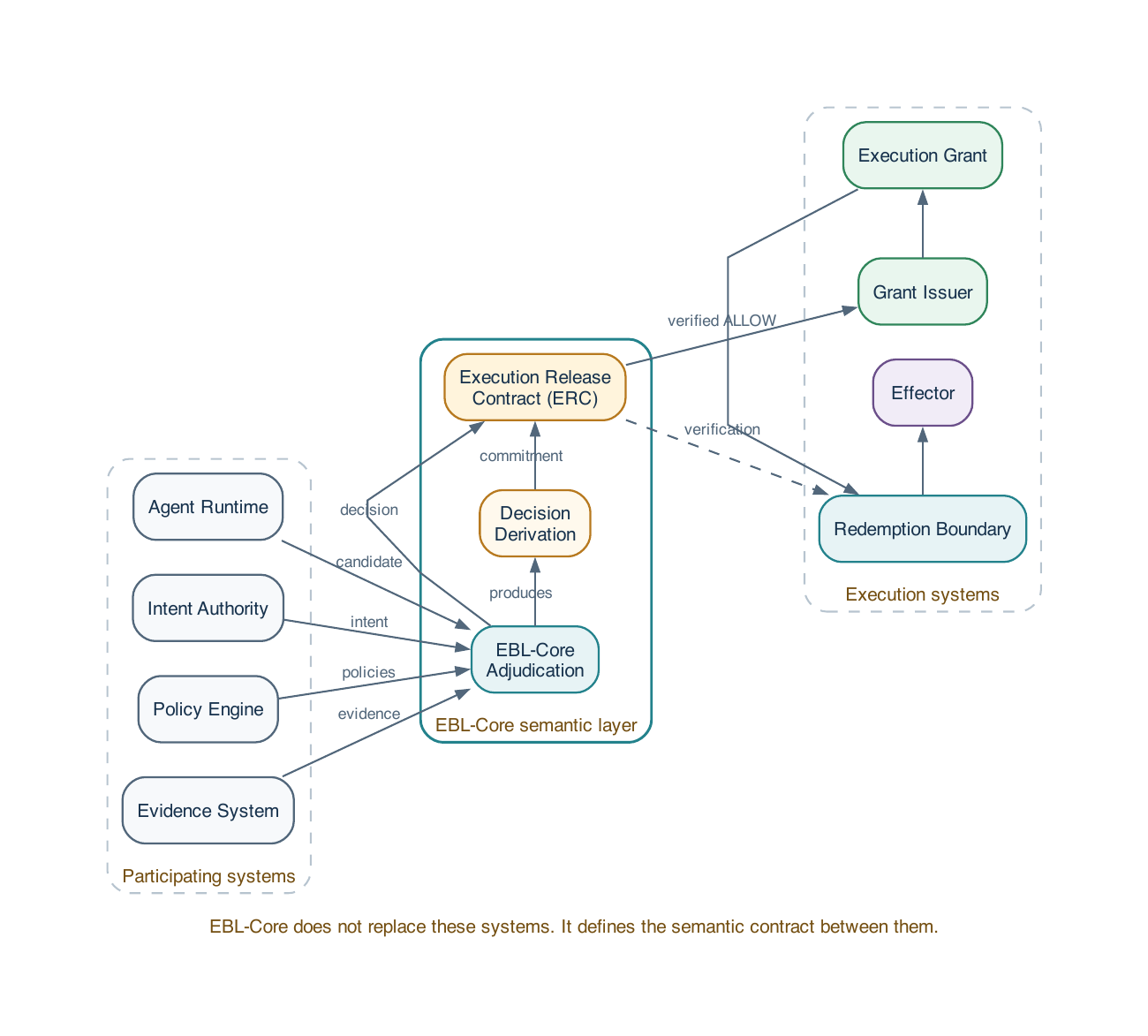}
  \caption{EBL-Core as an interoperability contract. Existing agent, intent, policy, evidence, grant, redemption, and effector components retain their own abstractions. EBL-Core specifies the semantic bindings exchanged among them; it does not replace the participating systems.}
  \label{fig:ebl-interoperability-contract}
\end{figure*}

\section{Evaluation and Validation Strategy}\label{evaluation-and-validation-strategy}

Evaluation of EBL-Core requires separating three distinct questions: whether the profile defines precise and testable semantics, whether implementations realize those semantics correctly and efficiently, and whether a particular deployment satisfies the assumptions under which an execution boundary is meaningful. These questions require different forms of evidence. Conformance tests can validate the semantic profile without establishing deployment security, while implementation benchmarks can characterize operational cost without proving complete mediation.

This preprint reports a bounded executable validation using the accompanying minimal reference artifact. The artifact instantiates one financial-transfer profile, a reference adjudicator, a separately implemented verifier and Semantic Replay path, canonical schemas, mutation vectors, and an in-memory linearizable grant store. Broader cross-domain evaluation, independently developed implementations, production benchmarks, and deployment experience remain future work.

\subsection{Evaluation Goals}\label{evaluation-goals}

The evaluation is organized around four research questions.

\textbf{Q1: Can EBL-Core precisely represent execution-release obligations for high-risk AI actions?}

This question concerns semantic coverage. An evaluation should determine whether EBL-Core can represent the intent, candidate-action, policy, evidence, context, temporal, and redemption constraints required by representative execution scenarios. A scenario should be considered covered only when its decision-relevant obligations are represented through typed profile objects or explicitly defined extension points. Encoding an obligation solely as uninterpreted free text does not establish semantic coverage.

\textbf{Q2: Can independent implementations produce equivalent decisions from the same semantic inputs?}

Given the same profile version and canonical input bundle, conforming adjudicators should produce the same decision and reason code. Their Decision Derivations need not be structurally or byte-identical. Each derivation must independently verify against the same inputs, decision, reason-code precedence, and profile rules. Equality of serialized derivations or derivation commitments is required only when a particular test vector supplies that representation as an explicit input, not as a general consequence of determinism.

\textbf{Q3: Does the Execution Release Contract capture obligations that are not otherwise standardized as one object by existing authorization and agent-control mechanisms?}

This question concerns the residual abstraction addressed by EBL-Core. The evaluation should not ask whether another mechanism is computationally capable of encoding an EBL-Core condition. A sufficiently general policy language can encode many such conditions. Instead, it should determine whether each condition is:

\begin{enumerate}
\def\labelenumi{\arabic{enumi}.}
\tightlist
\item
  native to the mechanism's documented abstraction;
\item
  expressible through an explicit adaptation;
\item
  delegated to an external component; or
\item
  not represented by the evaluated integration.
\end{enumerate}

The relevant result is therefore an obligation mapping, not a ranking of systems.

\textbf{Q4: Can EBL-Core identify execution-boundary and grant-lifecycle failures?}

The evaluation should test candidate, policy, evidence, context, time, derivation, ERC, and grant mutations. It should also test the lifecycle states \texttt{ISSUED}, \texttt{CONSUMED}, \texttt{EXPIRED}, and \texttt{REVOKED}, concurrent Redemption attempts, and the requirement that at most one successful Redemption linearize with the protected effect.

These questions evaluate EBL-Core as a semantic conformance profile. They do not measure whether an AI system is aligned, whether its proposed actions are desirable in general, or whether every path to an external effect is mediated.

\subsection{Preliminary Reference Artifact}\label{preliminary-reference-artifact}

The source package includes an executable artifact written against Python 3.13.4 using only the standard library. It provides two machine-readable schemas, profile-specific deterministic canonicalization and SHA-256 commitments, a reference adjudicator, ERC Generation and Grant Issuance, and an in-memory grant store whose lock is the logical linearization point for the protected test effect. A separate verifier recomputes structural validity, obligation statuses, decision and primary reason, derivation commitments, ERC bindings, and Semantic Replay without calling the adjudicator. The two paths share the declared canonicalization and commitment primitives, so they are not claimed as independently developed implementations. The implementation, schemas, test vectors, retained result, and reproduction instructions are available under \texttt{anc/} in the arXiv source package.

The corpus fixes one financial-transfer bundle and applies positive, negative, boundary, ordering, and adversarial mutations. Table~\ref{tab:preliminary-results} reports the retained run included with the source package. Every static vector checks the expected decision and primary reason, Decision-Derivation Verification, and Semantic Replay. Lifecycle checks cover mutation of a derivation and ERC, denial of Grant Issuance from a \texttt{DENY} ERC, candidate and policy-version changes, terminal states, and duplicate Redemption.

\begin{table*}[t]
\centering
\small
\setlength{\tabcolsep}{4pt}
\renewcommand{\arraystretch}{1.12}
\begin{tabular}{@{}>{\raggedright\arraybackslash}p{0.23\textwidth}>{\raggedright\arraybackslash}p{0.42\textwidth}>{\raggedright\arraybackslash}p{0.30\textwidth}@{}}
\toprule
\textbf{Validation component} & \textbf{Workload} & \textbf{Observed result} \\
\midrule
Static conformance corpus & 34 canonical, malformed, policy, intent, and evidence-state vectors & 34/34 matched decision, reason, derivation-verification, and replay oracles. \\
Lifecycle and mutation checks & 15 ERC, derivation, issuance, binding, expiry, revocation, and consumption checks & 15/15 matched the specified outcome. \\
Concurrent Redemption & 100 trials, each with 32 requests for one \texttt{ISSUED} grant (3,200 attempts) & Exactly one successful Redemption and one protected test effect per trial. \\
Revoke--Redeem race & 100 two-party races over one authoritative grant state & 100/100 ended in one valid terminal outcome with no effect after revocation. \\
\bottomrule
\end{tabular}
\caption{Preliminary executable validation for the included financial-transfer profile. Results establish behavior only for the artifact, environment, and test oracles described here.}
\label{tab:preliminary-results}
\end{table*}

During development, the corpus exposed a disagreement over whether an empty target was a structural error or an intent-binding failure; aligning both paths with the declared schema removed the disagreement. This is evidence that reason-code precedence and independent recomputation are testable rather than merely descriptive. The final retained run passed all reported checks. The artifact does not constitute mechanized verification, a production payment implementation, or evidence of complete mediation, and its local timing output is explicitly non-normative.

\subsection{Semantic Conformance Evaluation}\label{semantic-conformance-evaluation}

Let a test input be the closed semantic bundle

\[
B_v=\langle v,K,P,I,x,Ev,Ctx,t\rangle,
\]

where \(x\) is one canonical, fully materialized candidate and \(v\) is the EBL-Core version. An adjudicator evaluates the bundle as

\[
A(B_v)=\langle d,r,\pi\rangle.
\]

A conforming verifier checks:

\[
VerifyDerivation(B_v,d,r,\pi)=true.
\]

A conformance corpus should contain valid, invalid, boundary, and malformed bundles together with expected decisions, reason codes, derivation-verification conditions, ERC Generation outcomes, Grant Issuance outcomes, and Redemption outcomes.

\subsubsection{Input Closure}\label{input-closure}

For

\[
A(B_v)=\langle d,r,\pi\rangle
\]

and

\[
A(B'_v)=\langle d',r',\pi'\rangle,
\]

input closure and determinism require:

\[
Canon_v(B_v)=Canon_v(B'_v)
\Longrightarrow
d=d'\land r=r'.
\]

They do not require \(\pi=\pi'\). Instead, both supplied derivations must satisfy:

\[
VerifyDerivation(B_v,d,r,\pi)=true
\]

and

\[
VerifyDerivation(B'_v,d',r',\pi')=true.
\]

A conforming adjudicator must not obtain decision-relevant information from undeclared clocks, mutable process state, network lookups, or implementation-local defaults. Such information must first be materialized in the closed bundle.

Cross-implementation comparison should include:

\begin{itemize}
\tightlist
\item
  decision equality;
\item
  reason-code equality under the specified precedence rules;
\item
  acceptance of the same canonical inputs;
\item
  rejection of the same malformed inputs;
\item
  successful verification of every accepted Decision Derivation; and
\item
  compatible ERC verification and lifecycle outcomes.
\end{itemize}

The comparison must not require byte-identical Decision Derivations.

\subsubsection{Intent-binding tests}\label{intent-binding-tests}

Intent-binding vectors should include both permitted refinements and prohibited expansions.

Valid cases should cover:

\begin{itemize}
\tightlist
\item
  specialization of parameters within an authorized range;
\item
  selection of exactly one resource from a finite set explicitly encoded in the trusted intent, followed by materialization of that selected resource as the single canonical candidate;
\item
  reduction of an amount, duration, or effect scope;
\item
  selection of an explicitly permitted target; and
\item
  materialization of implementation details that do not alter the authorized effect.
\end{itemize}

Invalid cases should cover:

\begin{itemize}
\tightlist
\item
  substitution of a different recipient or target;
\item
  expansion of the affected resource set;
\item
  privilege escalation;
\item
  increase of financial, temporal, or physical effect;
\item
  substitution of an approval issued for another intent;
\item
  use of an approver outside the required authority scope; and
\item
  reinterpretation of an ambiguous intent in a more permissive direction.
\end{itemize}

A decision for one member of an intent-authorized set does not authorize an unspecified member at Redemption. Any mutation that changes the canonical candidate after ERC Generation requires a new adjudication.

\subsubsection{Evidence-State Tests}\label{evidence-state-tests}

For each obligation in \(Q_K\cup Q_P\), the corpus should exercise:

\[
\{\textsf{VALID},\textsf{UNKNOWN},\textsf{MISSING},
\textsf{EXPIRED},\textsf{CONFLICT}\}.
\]

Only \texttt{VALID} may discharge a positive obligation. \texttt{UNKNOWN}, \texttt{MISSING}, \texttt{EXPIRED}, and \texttt{CONFLICT} must not produce \texttt{ALLOW} for the candidate under adjudication. If a system wishes to propose a different reduced-risk or recovery operation, it must materialize that operation as a new canonical candidate and adjudicate it under the applicable policies and obligations.

The corpus should additionally test:

\begin{itemize}
\tightlist
\item
  evidence with the wrong type;
\item
  evidence supplied by an inadmissible provider;
\item
  evidence bound to another intent or candidate;
\item
  evidence at both sides of its validity boundary;
\item
  conflicting records from multiple providers;
\item
  equivalent evidence inputs in different representational orders;
\item
  undeclared evidence-resolver state; and
\item
  evidence that expires between adjudication and Redemption.
\end{itemize}

These tests evaluate obligation resolution and binding. They do not establish that an evidence provider's assertion is true.

\subsubsection{Policy-Composition Tests}\label{policy-composition-tests}

Policy evaluation should test three distinct properties.

First, \textbf{root-policy dominance} requires that Operational Policy neither alter the Root Policy predicate nor remove, weaken, rename, or reinterpret \(Q_K\). An Operational Policy permission cannot override a Root Policy denial.

Second, \textbf{conjunctive rule addition} should be tested by adding an Operational Policy restriction and confirming that the admitted candidate set does not expand. The resulting operational obligations must be added through \(Q_P\), without modifying \(Q_K\).

Third, \textbf{policy-version evolution} should be tested through the separately defined compatibility relation \(Compat_K(P_{old},P_{new})\). Version order, naming, or the presence of additional rules does not establish compatibility. Even when a change is shown to be non-expanding under fixed \(K\), compatibility does not automatically preserve an existing grant. Baseline ERC and grant reuse requires exact policy-version equality. A Root Policy change requires new adjudication.

The corpus should therefore include:

\begin{itemize}
\tightlist
\item
  an Operational Policy that attempts to permit a Root Policy denial;
\item
  deletion or reinterpretation of an obligation in \(Q_K\);
\item
  addition of an operational restriction and corresponding \(Q_P\);
\item
  an Operational Policy update that expands the permitted set;
\item
  a non-expanding update satisfying the accepted compatibility procedure;
\item
  an exact policy-version mismatch at ERC verification or Redemption;
\item
  an attempted reuse of a grant after a compatible Operational Policy update;
\item
  a Root Policy version change; and
\item
  policies exceeding the profile's declared structural or evaluation bounds.
\end{itemize}

Unsupported or out-of-bounds policies must be rejected explicitly rather than evaluated through implementation-dependent behavior.

\subsection{Adversarial Mutation Evaluation}\label{adversarial-mutation-evaluation}

Adversarial mutation testing begins with a valid trace:

\[
\tau=\langle B,d,r,\pi,erc,g,state(g)\rangle,
\]

where \(B\) contains one canonical candidate, the ERC verifies, and \(g\) is in the \texttt{ISSUED} state. A mutation changes one decision-relevant component while holding the others constant. The test oracle then evaluates the mutation at adjudication, ERC verification, Grant Issuance, or Redemption, according to when it occurs.

\begin{table*}[t]
\centering
\small
\setlength{\tabcolsep}{4pt}
\renewcommand{\arraystretch}{1.12}
\begin{tabular}{@{}>{\raggedright\arraybackslash}p{0.18\textwidth}>{\raggedright\arraybackslash}p{0.38\textwidth}>{\raggedright\arraybackslash}p{0.39\textwidth}@{}}
\toprule
\textbf{Mutation class} & \textbf{Example} & \textbf{Required response} \\
\midrule
Candidate mutation & A transfer to Alice is changed to a transfer to Bob. & Reject the existing decision, ERC, and grant binding; require new adjudication. \\
Evidence mutation & Evidence is replaced, expires, or is rebound after adjudication. & Reject the existing commitment or require new adjudication. \\
Policy mutation & A Root or Operational Policy version changes. & Reject on baseline version mismatch; \(Compat_K\) does not automatically preserve the grant. \\
Obligation mutation & \(Q_K\) is removed or reinterpreted, or \(Q_P\) changes. & Reject; obligation identity and resolution semantics are decision-relevant. \\
Context mutation & A balance, permission, device state, environment, or risk classification changes. & Deny when the committed version or current-state predicate no longer holds. \\
Derivation mutation & A step, premise, rule identifier, or conclusion in \(\pi\) changes. & Reject unless the resulting Decision Derivation independently verifies. \\
ERC substitution & An ERC is presented with a different candidate, bundle, issuer, or grant. & Reject the inconsistent binding chain. \\
Grant-state mutation & A \texttt{CONSUMED}, \texttt{EXPIRED}, or \texttt{REVOKED} grant is presented as usable. & Deny; terminal states cannot return to \texttt{ISSUED}. \\
Concurrent Redemption & Two requests redeem the same \texttt{ISSUED} grant concurrently. & At most one request may complete the linearized \texttt{ISSUED}-to-\texttt{CONSUMED} transition and protected effect. \\
Check-effect race & Decision-relevant state changes between validation and effect. & Deny unless validation, state transition, and effect can be linearized. \\
\bottomrule
\end{tabular}
\end{table*}

The baseline profile permits one successful Redemption for a grant instance. The mutation corpus must not use a reusable or bounded-use grant as its baseline. A second Redemption attempt returns \texttt{DENY}, including when the attempts overlap in time.

The principal metrics are:

\begin{itemize}
\tightlist
\item
  agreement with the expected decision or lifecycle outcome;
\item
  reason-code correctness;
\item
  detection stage;
\item
  whether re-adjudication is required;
\item
  false acceptance of a decision-relevant mutation;
\item
  false rejection of representations declared canonically equivalent; and
\item
  the number of successful effects produced by concurrent attempts against one grant.
\end{itemize}

\subsection{Baseline Abstraction Comparison}\label{baseline-comparison-methodology}

Table~\ref{tab:abstraction-matrix} provides the completed qualitative mapping for six close baselines. It uses the same native, adapted, external, and unestablished categories proposed by the evaluation design and gives each baseline credit for its documented primary abstraction. The result supports a limited positioning claim: the mandatory EBL-Core contract is not native as a whole to any compared abstraction. It does not show that those mechanisms are unable to express the conditions, nor does it establish that an EBL-Core implementation is more secure or efficient.

The current mapping is literature-based rather than adapter-based. A stronger evaluation requires concrete Cedar, Rego, capability, and agent-runtime adapters applied to the same conformance traces. Such work should record which guarantees are inherited from the substrate and which are supplied by wrapper code, then test cross-implementation ERC and Redemption outcomes. The included artifact does not yet perform this adapter comparison.

\subsection{Scenario-Based Evaluation}\label{scenario-based-evaluation}

Section~\ref{worked-transfer-example} and the accompanying corpus instantiate one complete financial-operation trace. Extending the same evaluation to infrastructure, deployment, disclosure, and physical-actuation profiles remains future work. Each additional scenario should provide:

\begin{enumerate}
\def\labelenumi{\arabic{enumi}.}
\tightlist
\item
  a structured intent;
\item
  one canonical, fully materialized candidate action;
\item
  versioned Root and Operational Policies;
\item
  separately identified \(Q_K\) and \(Q_P\);
\item
  complete typed evidence, explicit context, and explicit time;
\item
  an expected decision, reason code, and verifiable Decision Derivation;
\item
  the expected ERC Generation outcome;
\item
  Grant Issuance and Redemption outcomes for verified \texttt{ALLOW} cases; and
\item
  adjudication-time, ERC-time, and Redemption-time mutations.
\end{enumerate}

The broader corpus is intended to contain the following scenarios.

\begin{table*}[t]
\centering
\small
\setlength{\tabcolsep}{4pt}
\renewcommand{\arraystretch}{1.12}
\begin{tabular}{@{}>{\raggedright\arraybackslash}p{0.18\textwidth}>{\raggedright\arraybackslash}p{0.38\textwidth}>{\raggedright\arraybackslash}p{0.39\textwidth}@{}}
\toprule
\textbf{Scenario} & \textbf{Candidate-action binding} & \textbf{Representative obligations and failure probes} \\
\midrule
Financial operation & Exact source account, asset, amount, recipient, network, and transaction parameters & Approval authority and validity, balance evidence, amount limits, recipient substitution, fee or network changes, and duplicate redemption \\
Infrastructure change & Exact plan or configuration digest, target resources, environment, and intended transition & Change window, operator scope, pre-state commitment, rollback readiness, target expansion, and production-state drift \\
Software deployment & Exact artifact digest, service, release configuration, and destination environment & Artifact provenance, approval scope, environment constraints, canary or rollback conditions, artifact substitution, and approval expiry \\
Data disclosure & Exact data fields, recipient, channel, purpose, and disclosure operation & Recipient authority, data classification, scope limitation, context-dependent restrictions, field expansion, and channel substitution \\
Physical actuation & Exact command, device, parameters, duration, and actuation interface & Device state, safety envelope, operator authority, stale sensor evidence, parameter expansion, and delayed redemption \\
\bottomrule
\end{tabular}
\end{table*}

Each future scenario should include at least one permitted trace, denials for all relevant non-\texttt{VALID} obligation states, a Root Policy denial, an Operational Policy restriction, mutations after ERC Generation, and---where a grant is issued---lifecycle and concurrency tests after Grant Issuance. Ambiguities that cannot be represented without application-specific semantics should be recorded as profile gaps or extension requirements rather than resolved through unstated evaluator behavior.

Scenario coverage would show that the same conformance abstraction applies across several effect domains. It would not establish that EBL-Core captures every domain-specific safety property or guarantees a safe real-world outcome.

\subsection{Performance-Evaluation Scope}\label{performance-evaluation-for-a-future-artifact}

Performance is a property of an implementation rather than the semantic definition. The included artifact can emit indicative local timings, but the retained validation result and the claims in this paper do not depend on them. A standard-library Python model, one profile, and an in-memory store do not provide a meaningful production-performance baseline. Operational measurements for production-oriented implementations should therefore be reported separately from semantic conformance results.

The evaluation should vary policy size, the number of obligations in \(Q_K\) and \(Q_P\), evidence-set size, candidate complexity, context size, and Decision Derivation depth. It should separately measure canonicalization, adjudication, derivation construction, ERC Generation, cryptographic operations, Decision-Derivation Verification, Grant Issuance, and Redemption.

Such an evaluation should include:

\begin{itemize}
\tightlist
\item
  adjudication latency distributions;
\item
  serialized ERC and Decision Derivation sizes;
\item
  Decision-Derivation Verification latency;
\item
  Semantic Replay cost;
\item
  Grant Issuance and Redemption latency;
\item
  peak and steady-state memory use;
\item
  scaling with policy rules and Evidence Obligations; and
\item
  observed and analytically derived worst-case evaluation steps.
\end{itemize}

Decision-Derivation Verification checks whether a supplied derivation supports the committed result under the applicable rules. Semantic Replay reconstructs the decision from the committed semantic inputs. These operations should be measured independently.

External evidence acquisition should be excluded from adjudication latency unless it is explicitly part of the implementation under test. Performance comparisons should also match semantic work: an EBL-Core path that generates and verifies an ERC should not be compared directly with a baseline invocation that returns only a Boolean decision.

No universal latency or throughput threshold follows from EBL-Core. Acceptable operational bounds depend on the Effector domain, while termination and profile-defined evaluation limits remain conformance requirements.

\subsection{Artifact Status and Roadmap}\label{artifact-roadmap}

The current package completes a bounded subset of the artifact plan: machine-readable schemas, profile-specific canonicalization and commitments, one deterministic reference adjudicator, a separate verifier and Semantic Replay path, a mutation corpus, and lifecycle and concurrency tests for one financial-transfer profile. It is intended as an executable specification and test oracle, not as a production security component. Because the adjudicator and verifier share canonicalization and commitment primitives and were developed in the same artifact, the package does not claim implementation independence.

The next validation stages are independently developed adjudicators and verifiers, additional domain profiles, cross-implementation vectors, and adapters for representative authorization engines, agent runtimes, proof systems, and capability mechanisms. Each adapter should identify which required semantic properties are native, adapted, external, or unestablished. Performance characterization, mechanized correspondence between the formal rules and executable model, and production deployment remain separate stages requiring stronger evidence, Effector integration, authority analysis, and bypass evaluation.

\subsection{Limitations}\label{limitations}

The proposed evaluation cannot establish the following properties:

\begin{itemize}
\tightlist
\item
  that a natural-language instruction was translated into the correct trusted intent;
\item
  that an evidence provider's assertion is true;
\item
  that every execution path in a deployment passes through the designated boundary;
\item
  that the root policy is complete or substantively correct;
\item
  that arbitrary deployments provide complete mediation;
\item
  that issuer keys, policy authorities, or evidence providers are operationally uncompromised; or
\item
  that an authorized command produces the intended physical or external outcome.
\end{itemize}

Mutation testing can show that the specified boundary rejects tested classes of changed inputs. It cannot establish the absence of hidden Effectors, side channels, unmodeled state, or alternative authority paths. Decision-Derivation Verification and Semantic Replay can establish how a decision follows from committed inputs under the profile semantics; neither establishes that those inputs accurately represent the external world.

The evidence required for each claim category must therefore remain explicit:

\begin{table*}[t]
\centering
\small
\setlength{\tabcolsep}{4pt}
\renewcommand{\arraystretch}{1.12}
\begin{tabular}{@{}>{\raggedright\arraybackslash}p{0.18\textwidth}>{\raggedright\arraybackslash}p{0.38\textwidth}>{\raggedright\arraybackslash}p{0.39\textwidth}@{}}
\toprule
\textbf{Claim category} & \textbf{Appropriate validation evidence} & \textbf{What the evidence does not establish} \\
\midrule
Semantic conformance & Formal definitions, canonical vectors, differential evaluation, and derivation verification & Truth of inputs or complete mediation \\
Implementation behavior & Mutation tests, performance measurements, resource bounds, and verifier independence & Correctness of policies or external providers \\
Deployment security & Architecture review, effector integration, authority configuration, key management, and bypass analysis & Universal safety outside the evaluated deployment \\
\bottomrule
\end{tabular}
\end{table*}

Under these limitations, EBL-Core remains testable as a conformance profile. Its semantic properties can be evaluated through canonical inputs, decision and reason-code agreement, binding checks, Decision-Derivation Verification, lifecycle tests, and Semantic Replay. Implementation claims require executable artifacts and measurements, while deployment-security claims remain conditional on explicit authority, trust, Effector-fidelity, and mediation assumptions.

\section{Discussion}\label{discussion}

EBL-Core defines a conformance profile connecting adjudication to the release and Redemption of action-scoped execution authority. Its contribution is narrower than a complete authorization architecture or AI safety system. This section clarifies the distinct semantic roles in that profile and the deployment assumptions under which they are meaningful.

\begin{table*}[t]
\centering
\small
\setlength{\tabcolsep}{4pt}
\renewcommand{\arraystretch}{1.12}
\begin{tabular}{@{}>{\raggedright\arraybackslash}p{0.18\textwidth}>{\raggedright\arraybackslash}p{0.38\textwidth}>{\raggedright\arraybackslash}p{0.39\textwidth}@{}}
\toprule
\textbf{Claim area} & \textbf{EBL-Core defines} & \textbf{EBL-Core does not establish} \\
\midrule
Intent binding & A relation between a trusted intent object and one canonical candidate & Correct understanding of natural-language or human intent \\
Candidate identity & Commitments binding adjudication, ERC, grant, and Redemption to the same candidate & That the candidate is beneficial, complete, or error-free \\
Evidence & Obligation-relative states, bindings, freshness, and conflict handling & External truth, completeness, or provider honesty \\
Policy & Root-policy dominance, \(Q_K/Q_P\) separation, conjunctive addition, and version rules & Policy correctness or secure Root Policy administration \\
Decision & Deterministic decision and reason-code semantics with Decision-Derivation Verification & Correctness of unmodeled inputs or trusted authorities \\
ERC & A decision-binding release-condition object & An inherently authority-bearing token or proof of execution \\
Authority release & Grant Issuance from a verified \texttt{ALLOW} ERC & A new general-purpose capability primitive \\
Redemption & Single-use grant states and linearized validation, consumption, and protected effect & Universal mediation or faithful execution by every Effector \\
Outcome & A basis for relating release decisions to separately collected evidence & Proof that an external effect occurred or had the intended result \\
Security & Semantic properties required of conforming implementations under stated assumptions & Deployment security without authority, key, Effector, and bypass analysis \\
AI safety & A constrained interface between proposed actions and execution authority & General alignment, safe planning, or universal prevention of harmful actions \\
\bottomrule
\end{tabular}
\caption{Summary of EBL-Core claim boundaries. The profile specifies conditional semantic properties, while deployment and external-world claims require additional evidence.}
\label{tab:claim-boundaries}
\end{table*}

\subsection{EBL-Core Is Complementary to Authorization, Not a Replacement}\label{ebl-core-is-complementary-to-authorization-not-a-replacement}

Authorization mechanisms can evaluate both broad permission classes and highly specific structured requests. EBL-Core does not distinguish itself by assuming that authorization is limited to coarse-grained decisions. Instead, it asks whether an integration preserves a prescribed release-and-redemption contract for one canonical candidate.

The execution-boundary question is:

\begin{quote}
What minimum semantic release-and-redemption contract must hold before one canonical, fully materialized AI-generated candidate may receive action-scoped execution authority?
\end{quote}

Cedar, OPA/Rego, XACML, capability systems, and other mechanisms may supply policy evaluation or authority-management substrates. An EBL-Core integration may invoke them during adjudication and incorporate their results into a Decision Derivation.

The residual contract jointly binds the established intent, canonical candidate, Root and Operational Policy versions, \(Q_K\), \(Q_P\), evidence, context, time, decision, and Decision Derivation. It then distinguishes ERC Generation, Grant Issuance, and Redemption.

EBL-Core should consequently be evaluated by whether implementations preserve these bindings and lifecycle roles, not by whether it replaces the expressiveness of existing authorization systems.

\subsection{ERC, Execution Grant, and Redemption Are Distinct Semantic Roles}\label{erc-execution-grant-and-redemption-are-distinct-semantic-roles}

An \textbf{Execution Release Contract} is a decision-binding release-condition object. It commits to an adjudication result and the conditions under which that result may support the release and Redemption of authority. It is not inherently authority-bearing.

An \textbf{Execution Grant} represents the action-scoped execution authority released by a Grant Issuer from a verified \texttt{ALLOW} ERC. Under the baseline profile, it is bound to one canonical candidate, begins in \texttt{ISSUED}, and can support at most one successful Redemption.

\textbf{Redemption} is the governed operation that verifies the ERC, Decision Derivation, candidate binding, current conditions, authority scope, validity interval, and grant state. Successful Redemption linearizes validation, the protected effect, and the transition from \texttt{ISSUED} to \texttt{CONSUMED}.

A deployment may encode an ERC and an Execution Grant in the same transport envelope or cryptographic credential. A conforming implementation must nevertheless preserve their distinct semantic roles: the ERC records release conditions, the grant represents released authority, and Redemption governs its exercise.

An \texttt{ALLOW} ERC asserts that the committed candidate satisfied the profile's adjudication conditions. The ERC permits an independent verifier to check that assertion; it does not itself establish that Grant Issuance occurred. Possession of either an ERC or grant also does not establish successful Redemption, faithful execution, or an external outcome. Those claims require separate runtime and evidentiary support.

\subsection{Intent Limitations}\label{intent-limitations}

EBL-Core does not solve natural-language intent understanding. Its trusted semantic boundary begins after an Intent Authority has produced a structured intent object.

The Intent Authority may obtain that object from a human instruction, workflow definition, approved plan, organizational process, or another trusted source. Determining whether this translation accurately captures a person's actual intention is outside the EBL-Core adjudication model.

Under EBL-Core semantics, a candidate accepted for Redemption must have the same canonical identity as the candidate committed by the verified ERC and Execution Grant. If the Effector faithfully realizes that candidate and the protected effect is reachable only through the governed interface, the resulting operation corresponds to the committed candidate. Effector fidelity and exclusive mediation are deployment assumptions.

The corresponding semantic statement is:

\begin{quote}
If Redemption succeeds, the candidate accepted by the Redemption Interface is bound to the trusted intent object under the specified \texttt{Bound} relation.
\end{quote}

This does not establish that the intent object correctly represents a human's actual intention or that the Effector produced the intended external outcome.

For example, an intent object may correctly bind a payment to a named recipient while still containing a recipient selected through an erroneous upstream interpretation. EBL-Core can detect later substitution of that recipient, but it cannot determine that the original selection was mistaken unless such information is supplied through policy, evidence, or a corrected intent object.

This distinction is important in AI-agent safety analysis. Intent binding constrains execution relative to a trusted representation; it does not establish the semantic correctness of that representation.

\subsection{Evidence Limitations}\label{evidence-limitations}

EBL-Core specifies evidence identity, type, source admissibility, binding, freshness, validity intervals, conflict handling, and obligation resolution. It distinguishes \texttt{VALID}, \texttt{UNKNOWN}, \texttt{MISSING}, \texttt{EXPIRED}, and \texttt{CONFLICT}, and prohibits a non-\texttt{VALID} state from discharging a positive obligation in \(Q_K\cup Q_P\).

These semantics do not establish evidence truth. EBL-Core cannot determine that a provider is honest, a sensor is accurate, an approval was informed, or every relevant fact was observed. Such claims require provider trust, measurement integrity, authority analysis, and domain-specific validation.

Execution-lineage, provenance, and outcome-verification systems address the complementary question of what evidence supports claims about an execution and its result. EBL-Core instead classifies supplied evidence against declared obligations before authority release. A conforming implementation may consume compatible provenance or attestation records, but \texttt{VALID} means only that a record discharges a specified obligation under the declared resolution semantics.

\subsection{Deployment Assumptions and Complete Mediation}\label{deployment-assumptions-and-complete-mediation}

A semantic requirement and a deployment property must be distinguished.

The relevant semantic requirement is:

\begin{quote}
A conforming redemption requires a valid execution grant whose bindings and redemption conditions hold.
\end{quote}

The corresponding deployment property is:

\begin{quote}
Every path capable of producing the protected external effect requires such a grant.
\end{quote}

EBL-Core specifies the former. Establishing the latter requires analysis of the deployed architecture.

A conforming adjudicator and verifier do not prevent an agent from reaching an unmediated administrative interface, invoking an alternative tool, using an independently held credential, or communicating with a component that does not enforce the grant. Nor does the profile prove that all signing keys are protected or that every Intent Authority, Policy Authority, Evidence Provider, Grant Issuer, and Effector behaves correctly.

Complete mediation therefore depends on the placement and authority of the execution boundary. A deployment must identify all components capable of producing the protected effect, restrict alternative authority paths, and ensure that effectors validate grants before acting. Hardware-backed boundaries or isolated execution environments may strengthen these assumptions, but they do not follow from the EBL-Core semantics alone.

Claims about an EBL-Core deployment must consequently state:

\begin{itemize}
\tightlist
\item
  which effects are considered protected;
\item
  which effectors can produce those effects;
\item
  where grants are validated;
\item
  which components and keys are trusted;
\item
  which alternative authority paths have been excluded; and
\item
  which failures remain outside the model.
\end{itemize}

Global non-bypassability is not a claim of this work.

\subsection{Adjacent Analytical Boundaries}\label{adjacent-analytical-boundaries}

EBL-Core does not determine which components, credentials, administrators, or coalitions can reach a protected effect through ordinary, recovery, update, or alternative paths. That is an authority-topology and bypass-analysis problem. The Execution Grant is a semantic object in the declared release lifecycle; its presence does not prove that the deployed architecture makes the grant causally necessary.

Nor does EBL-Core establish execution lineage or terminal outcomes. Provenance and outcome-verification mechanisms may supply evidence consumed by adjudication or retained after Redemption, but their production, completeness, and truth properties require separate analysis. Conversely, those mechanisms do not determine whether a candidate satisfied the release contract at decision time.

The profile is self-contained at these interfaces: it defines the authority scope, evidence obligations, candidate, policies, ERC, grant, and Redemption predicates needed for its own conformance claims. Authority-topology and execution-lineage analyses can strengthen deployment evidence without becoming prerequisites for the core semantics.

\subsection{Research Implications}\label{research-implications}

If the execution-release abstraction proves useful across implementations, future work can extend the included schemas, adjudicator, verifier, and corpus into interoperable ERC encodings and adapters for established policy and capability systems. Independently developed implementations would permit stronger differential tests than the two code paths in the current artifact.

Mechanized semantics could examine determinism, root-policy dominance, evidence-state treatment, intent refinement, version compatibility, and lifecycle preservation. Domain profiles for infrastructure, disclosure, deployment, and physical actuation could refine evidence obligations and reason codes without weakening the mandatory bindings. Runtime or hardware-backed integration may strengthen key protection, mediation, and Effector assumptions; those mechanisms alter deployment assurance rather than the semantic contract.

\section{Conclusion}\label{conclusion}

AI agents increasingly move from generating information to proposing actions that can modify infrastructure, deploy software, transfer assets, disclose information, or actuate physical systems. Authorization systems, policy engines, runtime monitors, provenance mechanisms, and agent guardrails provide important foundations for governing such actions. The interfaces among these mechanisms, however, do not necessarily impose the same semantic conditions on the final transition from one candidate action to execution authority.

This paper defines EBL-Core, an execution-boundary conformance profile that binds a trusted intent object, one canonical and fully materialized candidate, versioned Root and Operational Policies, \(Q_K\) and \(Q_P\), typed evidence, decision-relevant context, explicit time, and a verifiable Decision Derivation through an Execution Release Contract. The ERC is a decision-binding release-condition object rather than inherently authority-bearing. A verified \texttt{ALLOW} ERC may support separate Grant Issuance, while Redemption governs whether the resulting action-scoped authority can be exercised.

EBL-Core defines semantic properties required of conforming implementations, including candidate binding, root-policy dominance, Evidence Obligation handling, deterministic adjudication, Decision-Derivation Verification, and single-use linearized Redemption. These properties do not establish correct human-intent interpretation, evidence truth, universal mediation, global non-bypassability, faithful Effector behavior, or correct external outcomes. Such claims remain conditional on the deployment architecture and stated trust assumptions.

Future work can extend the included executable specification with independently developed adjudicators and verifiers, ERC interoperability experiments, broader conformance suites, integration adapters, and additional domain-specific profiles. These artifacts may provide a foundation for evaluating execution-release semantics across heterogeneous AI-agent systems. The paper defines a semantic contract for when and why an AI-generated action may receive execution authority under explicit assumptions.

\balance
\bibliographystyle{ACM-Reference-Format}
\bibliography{bibliography/references}

\end{document}